\documentclass[11pt,envcountsame,envcountsect,runningheads]{llncs} 

\usepackage[a4paper,margin=2.52cm]{geometry}

\usepackage{latexsym} 
\usepackage{xspace} 
\usepackage{amssymb} 
\usepackage{amsmath}
\usepackage{stmaryrd} 
\usepackage{url} 
\usepackage{oldgerm} 
\usepackage{graphicx} 
\input{xy} 
\xyoption{all} 
\usepackage[ruled,vlined,linesnumbered]{algorithm2e}

\newcommand{\mand}{\sqcap}
\newcommand{\mor}{\sqcup}
\newcommand{\V}{\forall}
\newcommand{\E}{\exists}
\def\tuple#1{\langle#1\rangle} 

\newcommand{\ovl}[1]{\overline{#1}}

\newcommand{\EXPTIME}{\textsc{Exp\-Time}\xspace}
\newcommand{\NEXPTIME}{\textsc{NExp\-Time}\xspace}
\newcommand{\NtEXPTIME}{\textsc{N2Exp\-Time}\xspace}

\newcommand{\mT}{\mathcal{T}} 
\newcommand{\mR}{\mathcal{R}} 
\newcommand{\mA}{\mathcal{A}} 
\newcommand{\mI}{\mathcal{I}} 
\newcommand{\mE}{\mathcal{E}} 
\newcommand{\mS}{\mathcal{S}} 

\newcommand{\KB}{\mathit{KB}}

\newcommand{\nI}{\mathbb{I}} 
\newcommand{\nC}{\mathbb{C}} 
\newcommand{\nE}{\mathbb{E}}

\def\defeq{=}

\def\eqref#1{(\ref{#1})}

\newcommand{\myEnd}{\mbox{}\hfill$\Box$}
\newcommand{\koniec}{\mbox{}\hfill$\Box$}

\let\oldmarginpar\marginpar
\renewcommand\marginpar[1]{\oldmarginpar{\center{\footnotesize{\em #1}}}}

\newcommand{\comment}[1]{} 

\newcommand{\CPDLreg}{CPDL$_{reg}$\xspace}

\newcommand{\SHI}{$\mathcal{SHI}$\xspace}
\newcommand{\SROIQ}{$\mathcal{SROIQ}$\xspace}

\newcommand{\ALC}{$\mathcal{ALC}$\xspace} 
\newcommand{\SHIQ}{$\mathcal{SHIQ}$\xspace} 
\newcommand{\SHQ}{$\mathcal{SHQ}$\xspace} 
\newcommand{\SHOIQ}{$\mathcal{SHOIQ}$\xspace} 
\newcommand{\ALCI}{$\mathcal{ALCI}$\xspace}

\newcommand{\CSHIQ}{$C_\mathcal{SHIQ}$\xspace}

\newcommand{\CN}{\mathbf{C}} 
 
\newcommand{\RN}{\mathbf{R}} 
\newcommand{\IN}{\mathbf{I}} 

\newcommand{\closure}{\mathsf{closure}}

\newcommand{\Type}{\mathit{Type}}
\newcommand{\SType}{\mathit{SType}}
\newcommand{\Status}{\mathit{Status}}
\newcommand{\Label}{\mathit{Label}}
\newcommand{\CELabel}{\mathit{CELabel}}
\newcommand{\CELabelT}{\mathit{CELabelT}}
\newcommand{\CELabelR}{\mathit{CELabelR}}
\newcommand{\CELabelI}{\mathit{CELabelI}}
\newcommand{\StatePred}{\mathit{StatePred}}
\newcommand{\AfterTransPred}{\mathit{ATPred}}
\newcommand{\RFormulas}{\mathit{RFmls}}

\newcommand{\DSTimeStamp}{\mathit{DSTimeStamp}}
\newcommand{\ETimeStamp}{\mathit{ETimeStamp}}

\newcommand{\FmlsRC}{\mathit{FmlsRC}}
\newcommand{\FmlFB}{\mathit{FmlFB}}
\newcommand{\ILConstraints}{\mathit{ILConstraints}}

\newcommand{\State}{\mathsf{state}}
\newcommand{\NonState}{\mathsf{non\textrm{-}state}}
\newcommand{\Complex}{\mathsf{complex}}
\newcommand{\Simple}{\mathsf{simple}}

\newcommand{\True}{\mathsf{true}}
\newcommand{\False}{\mathsf{false}}
\newcommand{\Sat}{\mathsf{open}}
\newcommand{\Unsat}{\mathsf{closed}}
\newcommand{\Incomplete}{\mathsf{incomplete}}
\newcommand{\PExpanded}{\mathsf{p\textrm{-}expanded}}
\newcommand{\FExpanded}{\mathsf{f\textrm{-}expanded}}
\newcommand{\Unexpanded}{\mathsf{unexpanded}}
\newcommand{\Null}{\mathsf{null}}

\newcommand{\TUS}{\mathsf{testingClosedness}}
\newcommand{\CQF}{\mathsf{checkingFeasibility}}

\SetKwInput{Input}{Input}
\SetKwInput{Output}{Output}
\SetKwInput{GlobalData}{Global data}
\SetKwInput{Purpose}{Purpose}
\SetKw{Call}{call}
\SetKw{Set}{set}
\SetKw{Exit}{exit}
\SetKw{Break}{break}
\SetKw{Continue}{continue}

\SetKwFunction{NewSucc}{NewSucc} 
\SetKwFunction{FindProxy}{FindProxy} 
\SetKwFunction{ConToSucc}{ConToSucc} 
\SetKwFunction{SRTR}{SRTR}
\SetKwFunction{Apply}{Apply}
\SetKwFunction{ApplyTransRule}{ApplyTransRule}
\SetKwFunction{ApplyConvRule}{ApplyConvRule}
\SetKwFunction{Tableau}{Tableau}
\SetKwFunction{Expand}{Expand}
\SetKwFunction{UpdateStatus}{UpdateStatus}
\SetKwFunction{UpdateUnsatNodes}{UpdateUnsatNodes}
\SetKwFunction{PropagateStatus}{PropagateStatus}
\SetKwFunction{Kind}{Kind}
\SetKwFunction{AfterTrans}{AfterTrans}
\SetKwFunction{FullLabel}{FullLabel}
\SetKwFunction{CEFullLabelR}{CEFullLabelR}
\SetKwFunction{TUnsat}{TUnsat}
\SetKwFunction{TSat}{TSat}
\SetKwFunction{ToExpand}{ToExpand}
\SetKwFunction{AFormulas}{AFmls}

\SetKwFunction{NDFormulas}{NDFmls}
\SetKwFunction{Cnj}{Cnj}

\newcommand{\Ext}{\mathit{ext}}
\newcommand{\Trans}[1]{\mathit{trans}_\mR(#1)}

\newcommand{\IFDL}[1]{\textrm{IFDL}(#1)}

\newcommand{\Size}{\mathsf{N}}

\newcommand{\UPS}{(\textrm{UPS})\xspace}
\newcommand{\UPSa}{(\textrm{UPS$_1$})\xspace}
\newcommand{\UPSb}{(\textrm{UPS$_2$})\xspace}
\newcommand{\US}{(\textrm{US})\xspace}
\newcommand{\KCC}{(\textrm{KCC})\xspace}
\newcommand{\KCCd}{(\textrm{KCC$_1$})\xspace}
\newcommand{\KCCa}{(\textrm{KCC$_2$})\xspace}
\newcommand{\KCCb}{(\textrm{KCC$_3$})\xspace}
\newcommand{\KCCc}{(\textrm{KCC$_4$})\xspace}
\newcommand{\KCCcp}{(\textrm{KCC$_5$})\xspace}
\newcommand{\NUS}{(\textrm{NUS})\xspace}
\newcommand{\FS}{(\textrm{FS})\xspace}
\newcommand{\FSa}{(\textrm{FS$_1$})\xspace}
\newcommand{\FSb}{(\textrm{FS$_2$})\xspace}
\newcommand{\TP}{(\textrm{TP})\xspace}
\newcommand{\TF}{(\textrm{TF})\xspace}

\newtheorem{Explanation}{Explanation}
\newenvironment{explanation}{\begin{Explanation}\begin{em}}{\end{em}\end{Explanation}}

\newcommand{\sssection}[1]{\subsubsection{#1}}

\title{ExpTime Tableaux for the Description Logic SHIQ Based on Global~State~Caching and Integer~Linear~Feasibility~Checking}

\titlerunning{ExpTime Tableaux for SHIQ}

\author{Linh Anh Nguyen} 
\institute{ 
Institute of Informatics, University of Warsaw\\ 
Banacha 2, 02-097 Warsaw, Poland\\ 
\email{nguyen@mimuw.edu.pl}
}

\authorrunning{L.A. Nguyen}	 

\begin{document} 
\maketitle  
\sloppy 

\begin{abstract}
We give the first \EXPTIME (complexity-optimal) tableau decision procedure for checking satisfiability of a knowledge base in the description logic \SHIQ when numbers are coded in unary. Our procedure is based on global state caching and integer linear feasibility checking. 

\medskip

\end{abstract}


\section{Introduction}

Automated reasoning in description logics (DLs) has been an active research area for more than two decades. It is useful, for example, for the Semantic Web in engineering and querying ontologies~\cite{BaaderDLH}. One of basic reasoning problems in DLs is to check satisfiability of a knowledge base in a considered DL. Most of other reasoning problems in DLs are reducible to this one~\cite{BaaderSattler01}.
In this paper we study the problem of checking satisfiability of a knowledge base in the DL \SHIQ, which extends the basic DL \ALC with transitive roles, hierarchies of roles, inverse roles and quantified number restrictions. 

Tobies in his PhD thesis~\cite{TobiesThesis} proved that the problem is \EXPTIME-complete (even when numbers are coded in binary). On page~127 of~\cite{TobiesThesis} he wrote: {\em ``The previous \EXPTIME-completeness results for \SHIQ rely on the highly inefficient automata construction of Definition~4.34 used to prove Theorem~4.38 and, in the case of knowledge base reasoning, also on the wasteful pre-completion technique used to prove Theorem~4.42. Thus, we cannot expect to obtain an implementation from these algorithms that exhibits acceptable runtimes even on relatively ``easy'' instances. This, of course, is a prerequisite for using SHIQ in real-world applications.''}

Together with Horrocks and Sattler, Tobies therefore developed a tableau decision procedure for \SHIQ~\cite{HorrocksST00,TobiesThesis}. Later, in~\cite{HorrocksS07} Horrocks and Sattler gave a tableau decision procedure for \SHOIQ, and in~\cite{HorrocksKS06} Horrocks et al.\ gave a tableau decision procedure for \SROIQ. Both the logics \SHOIQ and \SROIQ are more expressive than \SHIQ, but they also have higher complexity (\SHOIQ is \NEXPTIME-complete and \SROIQ is \NtEXPTIME-complete).  
According to the recent survey~\cite{MishraK11}, the third generation ontology reasoners that support \SHIQ or \SROIQ are FaCT, FaCT++, RACER, Pellet, KAON2 and HermiT. The reasoners FaCT, FaCT++, RACER and Pellet are based on tableaux like the ones in~\cite{HorrocksST00,HorrocksS07,HorrocksKS06}, KAON2 is based on a translation into disjunctive Datalog~\cite{MotikS06}, and HermiT is based on hypertableaux~\cite{GlimmHM10}.
Traditional tableau decision procedures like the ones in~\cite{HorrocksST00,HorrocksS07,HorrocksKS06} can be implemented with various optimization techniques to increase efficiency. 
However, such procedures use backtracking to deal with disjunction ($\mor$) and ``or''-branching (e.g., the ``choice'' operation~\cite{HorrocksST00}) and despite advanced blocking techniques (e.g., pairwise anywhere blocking~\cite{HorrocksST00}), their complexities are non-optimal. In particular, the decision procedure for \SHIQ given in~\cite{HorrocksST00,TobiesThesis} has complexity \NtEXPTIME  when numbers are coded in unary.\footnote{When the procedure is improved by using anywhere pairwise blocking, the complexity will be \NEXPTIME.} 

By applying global caching~\cite{Pratt80,GoreN11}, together with colleagues (Gor{\'e}, Sza{\l}as and Dunin-K\c{e}plicz) we have developed complexity-optimal (\EXPTIME) tableau decision procedures for several modal and description logics with \EXPTIME complexity~\cite{GoreNguyen05tab,GoreNguyenTab07,GoreNguyen07clima,NguyenSzalas09ICCCI,NguyenSzalas-KSE09,NguyenS10FI,NguyenS10TCCI,NguyenS11SL,dkns2011}. 
In~\cite{GoreNguyenTab07,NguyenSzalas-KSE09,NguyenS11SL,dkns2011} analytic cuts are used to deal with inverse roles and converse modal operators. As cuts may be not efficient in practice, Gor{\'e} and Widmann developed the first cut-free \EXPTIME tableau decision procedures, based on global state caching, for the DL \ALCI (for the case without ABoxes) \cite{GoreW09} and CPDL~\cite{GoreW10}. 
We have applied global state caching for the modal logic \CPDLreg~\cite{nCPDLreg-long} and the DLs \ALCI~\cite{Nguyen-ALCI} and \SHI~\cite{SHI-ICCCI} for the case with ABoxes to obtain cut-free \EXPTIME tableau decision procedures for them.

As \SHIQ is a well-known and useful DL, developing a complexity-optimal tableau decision procedure for checking satisfiability of a knowledge base in \SHIQ is very desirable. The lack of such a procedure in a relatively long period up to the current work suggests that the problem is not easy at all. For example, in~\cite{GoreW09} Gor{\'e} and Widmann wrote {\em ``The extension to role hierarchies and transitive roles should not present difficulties, but the extension to include nominals and qualified number restrictions is not obvious to us.''.}

In this paper we give the first \EXPTIME tableau decision procedure for checking satisfiability of a knowledge base in the DL \SHIQ when numbers are coded in unary.\footnote{This corrects the claim that the complexity is measured using binary representation for numbers, which appeared in the previous versions of the current paper and the conference paper~\cite{SHIQ-ICCSAMA}.} Our procedure is based on global state caching and integer linear feasibility checking. 
Both of them are essential for our procedure in order to have the optimal complexity. Global state caching can be replaced by global caching plus cuts, which still guarantee the optimal complexity. However, we choose global state caching to avoid inefficient cuts although it makes our procedure more complicated. 

We are aware of only Farsiniamarj's master thesis~\cite{Farsiniamarj08} (written under supervision of Haarslev) as a work that directly combines tableaux with integer programming. Some related works~\cite{OK99,HM01,HTM01} are discussed in that thesis and we refer the reader there for details. In~\cite{Farsiniamarj08} Farsiniamarj presented a hybrid tableau calculus for the DL \SHQ (a restricted version of \SHIQ without inverse roles), which is based on the so-called atomic decomposition technique and combines arithmetic and logical reasoning. He stated that {\em ``The most prominent feature of this hybrid calculus is that it reduces reasoning about qualified number restrictions to integer linear programming. [\ldots] In comparison to other standard description logic reasoners, our approach demonstrates an overall runtime improvement of several orders of magnitude.''}~\cite{Farsiniamarj08}. On the complexity matter, Farsiniamarj wrote {\em ``the complexity of the algorithm seems to be characterized by a double-exponential function''} \cite[page 79]{Farsiniamarj08}. That is, his algorithm for \SHQ is not complexity-optimal. 
 
Combining global (state) caching and integer linear feasibility checking to obtain a complexity-optimal (\EXPTIME) tableau decision procedure for \SHIQ is not a simple exercise as one might think. Integer linear programming using the ``branch and bound'' method~\cite{BranchAndBound} is known since 1960, Pratt's idea on global caching for PDL~\cite{Pratt80} is known since 1980 (the term ``global caching'' is absent in~\cite{Pratt80}), a formal formulation of global caching for some modal and description logics~\cite{GoreNguyenTab07,GoreN11} is known since 2007 (if not earlier~\cite{GoreNguyen05tab}), but the first complexity-optimal (\EXPTIME) tableau decision procedure for \SHIQ for the case when numbers are coded in unary is only proposed now, in the current paper. As for experts on tableaux with global caching, it took us a few months to make the idea of the combination precise and one year since our paper on \SHI~\cite{SHI-ICCCI} to the current paper (on \SHIQ).\footnote{Of course, we were busy with teaching duties and other papers.} Our method of exploiting integer linear programming is different from Farsiniamarj's one~\cite{Farsiniamarj08}: in order to avoid nondeterminism, we only check feasibility but do not find and use solutions of the considered set of constraints as in~\cite{Farsiniamarj08}. 

The rest of this paper is structured as follows. In Section~\ref{section: prel} we recall the notation and semantics of \SHIQ and introduce an integer feasibility problem for DLs. In Section~\ref{section: calculus} we present our tableau decision procedure for \SHIQ. We start by introducing data structures and the tableau framework. We then present illustrative examples. After that we specify our tableau rules in detail. At the end of the section we present theoretical results about our tableau decision procedure. Those results are proved in Section~\ref{section: proofs}. We conclude this work in Section~\ref{section: conc}. 


\section{Preliminaries}
\label{section: prel}

\subsection{Notation and Semantics of \SHIQ}

Our language uses a finite set $\CN$ of {\em concept names}, a finite set $\RN$ of {\em role names}, and a finite set $\IN$ of {\em individual names}. 
We use letters like $A$ and $B$ for concept names, $r$ and $s$ for role names, and  $a$ and $b$ for individual names. We refer to $A$ and $B$ also as {\em atomic concepts}, and to $a$ and $b$ as {\em individuals}. 

For $r \in \RN$, let $r^-$ be a new symbol, called the {\em inverse} of $r$. 
Let $\RN^{-} \defeq \{r^{-} \mid$ $r \in \RN\}$ be the set of {\em inverse roles}. For $r \in \RN$, define $(r^-)^- \defeq r$. A {\em role} is any member of $\RN \cup \RN^{-}$. We use letters like $R$ and $S$ to denote roles.  

An (\SHIQ) {\em RBox} $\mR$ is a finite set of role axioms of the form $R \sqsubseteq S$ or $R \circ R \sqsubseteq R$. 
By $\Ext(\mR)$ we denote the least extension of $\mR$ such that:
\begin{itemize}
\item $R \sqsubseteq R \in \Ext(\mR)$ for any role $R$
\item if $R \sqsubseteq S \in \Ext(\mR)$ then $R^- \sqsubseteq S^- \in \Ext(\mR)$
\item if $R \circ R \sqsubseteq R \in \Ext(\mR)$ then $R^- \circ R^- \sqsubseteq R^- \in \Ext(\mR)$
\item if $R \sqsubseteq S \in \Ext(\mR)$ and $S \sqsubseteq T \in \Ext(\mR)$ then $R \sqsubseteq T \in \Ext(\mR)$.
\end{itemize}

By $R \sqsubseteq_\mR S$ we mean $R \sqsubseteq S \in \Ext(\mR)$, and by $\Trans{R}$ we mean $(R \circ R \sqsubseteq R) \in \Ext(\mR)$. If $R \sqsubseteq_\mR S$ then $R$ is a~{\em subrole} of $S$ (w.r.t.~$\mR$).
If $\Trans{R}$ then $R$ is a~{\em transitive role} (w.r.t.~$\mR$). 
A~role is {\em simple} (w.r.t.~$\mR$) if it is neither transitive nor has any transitive subroles (w.r.t.~$\mR$).

{\em Concepts} in \SHIQ\ are formed using the following BNF grammar, where $n$ is a nonnegative integer and $S$ is a simple role:
\[
C, D ::=
        \top
        \mid \bot
        \mid A
        \mid \lnot C
        \mid C \mand D
        \mid C \mor D
        \mid \E R.C
        \mid \V R.C
	\mid\ \geq\!n\,S.C
	\mid\ \leq\!n\,S.C
\]

We use letters like $C$ and $D$ to denote arbitrary concepts.

A {\em TBox} is a~finite set of axioms of the form $C
\sqsubseteq D$ or $C \doteq D$.

An {\em ABox} is a~finite set of {\em assertions} of the form $a\!:\!C$, $R(a,b)$ or $a \not\doteq b$.

A {\em knowledge base} in \SHIQ\ is a tuple $\tuple{\mR,\mT,\mA}$, where $\mR$ is an RBox, $\mT$ is a~TBox and $\mA$ is an ABox.

We say that a role $S$ is {\em numeric} w.r.t. a knowledge base $\KB = \tuple{\mR,\mT,\mA}$ if: 
\begin{itemize}
\item it is simple w.r.t.~$\mR$ and occurs in a concept $\geq\!n\,S.C$ or $\leq\!n\,S.C$ in $\KB$, or 
\item $S \sqsubseteq_\mR R$ and $R$ is numeric w.r.t.~$\KB$, or
\item $S^-$ is numeric w.r.t.~$\KB$. 
\end{itemize}
We will simply call such an $S$ a numeric role when $\KB$ is clear from the context. 

A {\em formula} is defined to be either a concept or an ABox assertion. 
We use letters like $\varphi$, $\psi$ and $\xi$ to denote formulas.
%
Let $\Null\!:\!C$ stand for $C$. We use $\alpha$ to denote either an individual or $\Null$. Thus, $\alpha\!:\!C$ is a formula of the form $a\!:\!C$ or $\Null\!:\!C$ (which means $C$). 

An {\em interpretation} $\mI = \langle \Delta^\mI, \cdot^\mI
\rangle$ consists of a~non-empty set $\Delta^\mI$, called the {\em
domain} of $\mI$, and a~function $\cdot^\mI$, called the {\em
interpretation function} of $\mI$, that maps each concept name
$A$ to a~subset $A^\mI$ of $\Delta^\mI$, each role name
$r$ to a~binary relation $r^\mI$ on $\Delta^\mI$, and each individual name $a$ to an element $a^\mI \in \Delta^\mI$.
The interpretation function is extended to inverse roles and complex concepts as
follows, where $\sharp Z$ denotes the cardinality of a set~$Z$: 
\[
\begin{array}{c}
(r^-)^\mI = \{ \tuple{x,y} \mid \tuple{y,x} \in r^\mI\} 
\quad\quad\quad\quad
	\top^\mI = \Delta^\mI
\quad\quad\quad\quad
	\bot^\mI = \emptyset
\\[1.0ex]
(\lnot C)^\mI = \Delta^\mI - C^\mI
\quad\quad 
	(C \mand D)^\mI = C^\mI \cap D^\mI
\quad\quad 
	(C \mor D)^\mI = C^\mI \cup D^\mI
\\[1.0ex]
(\E R.C)^\mI =
        \big\{ x \in \Delta^\mI \mid \E y\big[ \tuple{x,y} \in R^\mI
                 \textrm{ and } y \in C^\mI\big]\big\}
\\[1.0ex]
(\V R.C)^\mI =
        \big\{ x \in \Delta^\mI \mid \V y\big[ \tuple{x,y} \in R^\mI
                       \textrm{ implies } y \in C^\mI\big]\big\}
\\[1.0ex]
(\geq\!n\,R.C)^\mI =
        \big\{ x \in \Delta^\mI \mid \sharp\{y \mid \tuple{x,y} \in R^\mI
                 \textrm{ and } y \in C^\mI \} \geq n \big\}
\\[1.0ex]
(\leq\!n\,R.C)^\mI =
        \big\{ x \in \Delta^\mI \mid \sharp\{y \mid \tuple{x,y} \in R^\mI
                 \textrm{ and } y \in C^\mI \} \leq n \big\}.
\end{array}
\]

Note that $(r^-)^\mI = (r^\mI)^{-1}$ and this is compatible with $(r^-)^- = r$.

For a set $\Gamma$ of concepts, define $\Gamma^\mI = \{x \in \Delta^\mI \mid x \in C^\mI \textrm{ for all } C \in \Gamma\}$.

The relational composition of binary relations $\mathrm{R}_1$, $\mathrm{R}_2$ is denoted by $\mathrm{R}_1 \circ \mathrm{R}_2$.

An interpretation $\mI$ is a~{\em model of an RBox} $\mR$ if for every axiom $R \sqsubseteq S$ (resp.\ $R \circ R \sqsubseteq R$) of $\mR$, we have that $R^\mI \subseteq S^\mI$ (resp.\ $R^\mI \circ R^\mI \subseteq R^\mI$). 
Note that if $\mI$ is a model of $\mR$ then it is also a model of $\Ext(\mR)$.

An interpretation $\mI$ is a~{\em model of a~TBox} $\mT$ if for
every axiom $C \sqsubseteq D$ (resp.\ $C \doteq D$) of $\mT$,
we have that $C^\mI \subseteq D^\mI$ (resp.\ $C^\mI = D^\mI$).

An interpretation $\mI$ is a~{\em model of an ABox} $\mA$ if for every assertion $a\!:\!C$ (resp.\ $R(a,b)$ or $a \not\doteq b$) of $\mA$, we have that $a^\mI \in C^\mI$ (resp.\ $\tuple{a^\mI,b^\mI} \in R^\mI$ or $a^\mI \neq b^\mI$).

An interpretation $\mI$ is a~{\em model of a~knowledge base}
$\tuple{\mR,\mT,\mA}$ if $\mI$ is a~model of all $\mR$, $\mT$ and~$\mA$.
A knowledge base $\tuple{\mR,\mT,\mA}$ is {\em satisfiable} if it has
a~model.

An interpretation $\mI$ {\em satisfies} a concept $C$ (resp.\ a set $X$ of concepts) if $C^\mI \neq \emptyset$ (resp.\ $X^\mI \neq \emptyset$). 
It {\em validates} a concept $C$ if $C^\mI = \Delta^\mI$. 
A set $X$ of concepts is {\em satisfiable w.r.t.\ an RBox $\mR$ and a TBox $\mT$} if there exists a model of $\mR$ and $\mT$ that satisfies $X$. 
For $X = \mA \cup \mA'$, where $\mA$ is an ABox and $\mA'$ is a set of assertions of the form $\lnot R(a,b)$ or $a \doteq b$, we say that $X$ is {\em satisfiable w.r.t.\ an RBox $\mR$ and a TBox $\mT$} if there exists a~model $\mI$ of $\tuple{\mR,\mT,\mA}$ such that: $\tuple{a^\mI,b^\mI} \notin R^\mI$ for all $(\lnot R(a,b)) \in \mA'$, and $a^\mI = b^\mI$ for all $(a \doteq b) \in \mA'$. In that case, we also say that $\mI$ is a model of $\tuple{\mR,\mT,X}$. 


\subsection{An Integer Feasibility Problem for Description Logics}

For dealing with number restrictions in \SHIQ, we consider the following integer feasibility problem:
\[
\begin{array}{c}
\displaystyle \sum_{j=1}^m a_{i,j}\cdot x_j \,\bowtie_i\, b_i,\  
	\textrm{ for } 1 \leq i \leq l;\\[1.5ex]
x_j \geq 0,\ 
	\textrm{ for } 1 \leq j \leq m;
\end{array}
\]
where each $a_{i,j}$ is either 0 or 1, each $x_j$ is a variable standing for a natural number, each $\bowtie_i$ is either $\leq$ or $\geq$, each $b_i$ is a natural number encoded by using no more than $n$ bits (i.e., $b_i \leq 2^n$). We call this an $\IFDL{l,m,n}$-problem (a~problem of Linear Integer Feasibility for Description Logics with size specified by $l,m,n$). The problem is {\em feasible} if it has a solution (i.e., values for the variables $x_j$, $1 \leq j \leq l$, that are natural numbers satisfying the constraints), and is {\em infeasible} otherwise. By solving an $\IFDL{l,m,n}$-problem we mean checking its feasibility. 

It is known from linear programming that, if the variables $x_j$ are not required to be natural numbers but can be real numbers then the above feasibility problem can be solved in polynomial time in $l$, $m$ and $n$. The general integer linear optimization problem is known to be NP-hard.\footnote{\url{http://en.wikipedia.org/wiki/Integer_programming}}

To solve an integer feasibility problem, we propose to use the decomposition technique and the ``branch and bound'' method~\cite{BranchAndBound}. One can first analyze dependencies between the variables and the constraints to decompose the problem into smaller independent subproblems, then solve the subproblems that are trivial, and after that apply the ``branch and bound'' method~\cite{BranchAndBound} to the remaining subproblems. 

The above mentioned approach may not guarantee that a given $\IFDL{l,m,n}$-problem is solved in exponential time in~$n$. We give below an estimation of the upper bound for the complexity for some specific cases, using another approach. 

\newcommand{\LemmaIFDL}{Every $\IFDL{l,m,n}$-problem such that 
$l \leq n$, 
$m$ is (at most) exponential in $n$, 
and $b_i \leq n$ for all $1 \leq i \leq l$ 
can be solved in (at most) exponential time in~$n$.}
\begin{lemma}\label{lemma: IFDL}
\LemmaIFDL
\end{lemma}

\begin{proof}
Consider the following nondeterministic procedure:
\begin{enumerate}
\item initialize $c_{i,j} := 0$ for each $1 \leq i \leq l$ and $1 \leq j \leq m$ such that $a_{i,j} = 1$
\item for each $i$ from 1 to $l$ do
  \begin{itemize}
  \item[] for each $k$ from 1 to $b_i$ do
     \begin{itemize}
     \item[] choose some $j$ among $1, \ldots, m$ such that $a_{i,j} = 1$ and set $c_{i,j} := c_{i,j} + 1$
     \end{itemize}
  \end{itemize}
\item if the set of constraints $\{x_j \,\bowtie_i\, c_{i,j} \mid 1 \leq i \leq l, 1 \leq j \leq m, a_{i,j} = 1\}$ is feasible then return ``yes'', else return ``no''.
\end{enumerate}

Observe that the considered $\IFDL{l,m,n}$-problem is feasible iff there exists a run of the above procedure that returns ``yes''. Since $b_i \leq n$ for all $1 \leq i \leq l$, there are no more than $m^{l\cdot n}$ possible runs of the above procedure. All the steps of the procedure can be executed in time $O(l\cdot m \cdot n)$. Since $l \leq n$ and $m$ is (at most) exponential in $n$, we conclude that the considered $\IFDL{l,m,n}$-problem can deterministically be solved in (at most) exponential time in~$n$.
\myEnd
\end{proof}

The following lemma is more general than the above lemma. 

\begin{lemma}\label{lemma: IFDL2}
Every $\IFDL{l,m,n}$-problem satisfying the following properties can be solved in (at most) exponential time in~$n\,$: 
\begin{itemize}
\item $l \leq n$, $m$ is (at most) exponential in $n$, 
\item and 
   \begin{itemize}
   \item either $b_i \leq n$ for all $1 \leq i \leq l$ such that $\bowtie_i$ is $\leq$
   \item or $b_i \leq n$ for all $1 \leq i \leq l$ such that $\bowtie_i$ is $\geq$.
   \end{itemize}
\end{itemize}
\end{lemma}

\begin{proof}
Suppose $l \leq n$, $m$ is (at most) exponential in $n$, and $b_i \leq n$ for all $1 \leq i \leq l$ such that $\bowtie_i$ is $\leq$. The other case is similar and omitted. Consider the following nondeterministic procedure:
\begin{enumerate}
\item let $J = \{ j \mid 1 \leq j \leq m$ and there exists $1 \leq i \leq l$ such that $\bowtie_i$ is $\leq$ and $a_{i,j} = 1\}$
\item for each $1 \leq i \leq l$ and $1 \leq j \leq m$ such that $\bowtie_i$ is $\leq$ and $a_{i,j} = 1$, set $c_{i,j} := 0$
\item for each $i$ from 1 to $l$ such that $\bowtie_i$ is $\leq$, do
  \begin{itemize}
  \item[] for each $k$ from 1 to $b_i$ do
     \begin{itemize}
     \item[] choose some $j$ among $1, \ldots, m$ such that $a_{i,j} = 1$ and set $c_{i,j} := c_{i,j} + 1$
     \end{itemize}
  \end{itemize}
\item for each $j \in J$ do
  \begin{itemize}
  \item[] $d_j := \min\{c_{i,j} \mid 1 \leq i \leq l,\; \bowtie_i$ is $\leq$ and $a_{i,j} = 1\}$
  \end{itemize}
\item if the set of constraints $\{\sum_{j=1}^m a_{i,j}\cdot x_j \geq b_i \mid 1 \leq i \leq l,\; \bowtie_i$ is $\geq\} \cup \{x_j = d_j \mid j \in J\}$ is feasible then return ``yes'', else return ``no''.
\end{enumerate}

Observe that the considered $\IFDL{l,m,n}$-problem is feasible iff there exists a run of the above procedure that returns ``yes''. Under the assumptions of the lemma, there are no more than $m^{l\cdot n}$ possible runs of the above procedure. All the steps of the procedure can be executed in time $O(l\cdot m \cdot n)$. Since $l \leq n$ and $m$ is (at most) exponential in $n$, we conclude that the considered $\IFDL{l,m,n}$-problem can deterministically be solved in (at most) exponential time in~$n$.
\myEnd
\end{proof}


\section{\EXPTIME Tableaux for \SHIQ}
\label{section: calculus}

\subsection{Data Structures and the Tableau Framework}

We assume that concepts and ABox assertions are represented in
negation normal form (NNF), where $\lnot$ occurs only directly
before atomic concepts.\footnote{Every formula can be
transformed to an equivalent formula in NNF in polynomial time.} 
We use $\overline{C}$ to denote the NNF of $\lnot C$, and for $\varphi = (a\!:\!C)$, we use $\ovl{\varphi}$ to denote $a\!:\!\ovl{C}$.
For simplicity, we treat axioms
of $\mT$ as concepts representing global assumptions: an axiom
$C \sqsubseteq D$ is treated as $\overline{C} \mor D$, while an
axiom $C \doteq D$ is treated as $(\overline{C} \mor D) \mand
(\overline{D} \mor C)$. That is, we assume that $\mT$ consists
of concepts in NNF. A~concept $C \in \mT$ can be thought of as an axiom $\top \sqsubseteq C$. Thus, an interpretation $\mI$ is a~model of
$\mT$ iff $\mI$ validates every concept $C \in \mT$. 
%

We define tableaux as rooted graphs. Such a graph is a tuple $G = \tuple{V,E,\nu}$, where $V$ is a set of nodes, $E \subseteq V \times V$ is a set of edges, $\nu \in V$ is the root, and each node $v \in V$ has a number of attributes. If there is an edge $\tuple{v,w} \in E$ then we call $v$ a {\em predecessor} of $w$, and call $w$ a {\em successor} of~$v$. 
Attributes of tableau nodes are:
\begin{itemize}
\item $\Type(v) \in \{\State, \NonState\}$. If $\Type(v) = \State$ then we call $v$ a {\em state}, else we call $v$ a {\em non-state} (or an {\em internal} node). If $\Type(v) = \State$ and $\tuple{v,w} \in E$ then $\Type(w) = \NonState$. If $\tuple{v,w} \in E$ and $\Type(w) = \State$ then $v$ has only $w$ as a successor.

\item $\SType(v) \in \{\Complex, \Simple\}$ is called the subtype of $v$. If $\SType(v) = \Complex$ then we call $v$ a {\em complex node}, else we call $v$ a {\em simple node}. The graph never contains edges from a simple node to a complex node. If $\tuple{v,w}$ is an edge from a complex node $v$ to a simple node $w$ then $\Type(v) = \State$ and $\Type(w) = \NonState$. If $\Type(v) = \State$ and $\tuple{v,w} \in E$ then $\SType(w) = \Simple$. The root of the graph is a complex node. 

\item $\Status(v) \in \{\Unexpanded$, $\PExpanded$, $\FExpanded$, $\Incomplete$, $\Unsat$, $\Sat\}$, where $\PExpanded$ and $\FExpanded$ mean ``partially expanded'' and ``fully expanded'', respectively. $\Status(v)$ may be $\PExpanded$ or $\Incomplete$ only when $\Type(v) = \State$. The status of a node may only be changed:
  \begin{itemize}
  \item from $\Unexpanded$ to any other status; or 
  \item from $\PExpanded$ to $\FExpanded$, $\Incomplete$ or $\Unsat$; or  
  \item from $\FExpanded$ to $\Incomplete$, $\Unsat$ or $\Sat$. 
  \end{itemize}
The values $\Incomplete$, $\Unsat$ and $\Sat$ are thus ``final''. Roughly speaking, $\Unsat$ means ``unsatisfiable'' and $\Sat$ means ``satisfiable'' in a certain sense. 

\item $\Label(v)$ is a finite set of formulas, called the label of $v$. The label of a complex node consists of ABox assertions, while the label of a simple node consists of concepts.

\item $\RFormulas(v)$ is a finite set of formulas, called the set of reduced formulas of~$v$.

\item $\StatePred(v) \in V \cup \{\Null\}$ is called the state-predecessor of $v$. It is available only when $\Type(v) = \NonState$. 
If $v$ is a non-state and $G$ has no paths connecting a state to $v$ then $\StatePred(v) = \Null$. Otherwise, $G$ has exactly one state $u$ that is connected to $v$ via a path not containing any other states, and we have that $\StatePred(v) = u$.

\item $\AfterTransPred(v) \in V$ is called the after-transition-predecessor of~$v$. It is available only when $\Type(v) = \NonState$ and $\SType(v) = \Simple$. In that case, if $v_0 = \StatePred(v)$ ($\neq \Null$) then there is exactly one successor $v_1$ of $v_0$ such that every path connecting $v_0$ to $v$ must go through $v_1$, and we have that $\AfterTransPred(v) = v_1$. 
We define $\AfterTrans(v) = (\AfterTransPred(v) = v)$. 
If $\AfterTrans(v)$ holds then $v$ has exactly one predecessor $u$ and $u$ is a state. 

\item $\CELabel(v)$ is a tuple $\tuple{\CELabelT(v),\CELabelR(v),\CELabelI(v)}$ called the coming edge label of $v$. It is available only when $\AfterTrans(v)$ holds. In that case:
  \begin{itemize}
  \item $\CELabelT(v) \in \{\TUS,\CQF\}$ (the letter $T$ stands for ``type''), 
  \item $\CELabelR(v) \subseteq \RN \cup \RN^{-}$ (is a set of roles), 
  \item $\CELabelI(v) \in \IN \cup \{\Null\}$ (for $u = \StatePred(v)$ being the only predecessor of~$v$, if $\SType(u) = \Complex$ then $\CELabelI(v)$ is an individual, else $\CELabelI(v) = \Null$).
  \end{itemize}

\item $\FmlsRC(v)$ is a set of formulas, called the set of formulas required by converse (inverse) for~$v$. It is available only when $\Type(v) = \State$. 

\item $\FmlFB(v)$ is a formula, called the formula for branching on at $v$. It is available only when $\Type(v) = \State$ and $\Status(v) = \Incomplete$. 

\item $\ILConstraints(v)$ is a set of integer linear constraints. It is available only when $\Type(v) = \State$. The constraints use variables $x_w$ indexed by successors $w$ of $v$ with $\CELabelT(w) = \CQF$. Such a variable $x_w$ specifies how many copies of $w$ will be used as successors of~$v$. 
\end{itemize}

We will use also new concept constructors $\preceq\!n\,R.C$ and $\succeq\!n\,R.C$, where $R$ is a numeric role. The difference between $\preceq\!n\,R.C$ and $\leq\!n\,R.C$ is that, for checking $\preceq\!n\,R.C$, we do not have to look to predecessors of the node. The aim of $\succeq\!n\,R.C$ is similar. We use $\preceq\!n\,R.C$ and $\succeq\!n\,R.C$ only as syntactic representations of some expressions, and do not provide semantics for them. 
We define 
\[
\aligned \FullLabel(v) =\; & \Label(v) \cup \RFormulas(v)\, -\\
& \{\textrm{formulas of the form } \alpha\!:(\preceq\!n\,R.C) \textrm{ or } \alpha\!:(\succeq\!n\,R.C)\}.\endaligned
\]

By the {\em local graph} of a state $v$ we mean the subgraph of $G$ consisting of all the paths starting from $v$ and not containing any other states. 
Similarly, by the local graph of a non-state $v$ we mean the subgraph of $G$ consisting of all the paths starting from $v$ and not containing any states.

\begin{remark}
We give here a further description/intuition about the structure of $G$:
\begin{itemize}
\item If $u$ is a state of $G$ and $v_1,\ldots,v_k$ are all successors of $u$ then: 
  \begin{itemize}
  \item the local graph of each $v_i$ is a directed acyclic graph,
  \item if $i,j \in \{1,\ldots,k\}$ and $i \neq j$ then the local graphs of $v_i$ and $v_j$ are disjoint,
  \item the local graph of $u$ is a graph rooted at $u$ and consisting of the edges from $u$ to $v_1,\ldots,v_k$ and the local graphs of $v_1,\ldots,v_k$,
  \item if $w$ is a node in the local graph of some $v_i$ then $\StatePred(w) = u$ and $\AfterTransPred(w) = v_i$.
  \end{itemize}

\item If $u$ is a state of $G$ then:
  \begin{itemize}
  \item each edge from outside the local graph of $u$ to the local graph of $u$ must end at~$u$,
  \item each edge outgoing from the local graph of $u$ must start from a non-state and is the only outgoing edge of that non-state.
  \end{itemize}

\item $G$ consists of:
  \begin{itemize}
  \item the local graph of the root $\nu$,
  \item the local graphs of states,
  \item edges coming to states.
  \end{itemize}

\item Each complex node of $G$ is like an ABox (more formally: its label is an ABox), which can be treated as a graph whose vertices are named individuals. On the other hand, a simple node of $G$ stands for an unnamed individual. If there is an edge from a simple non-state $v$ to (a simple node) $w$ then $v$ and $w$ stand for the same unnamed individual. An edge from a complex node $v$ to a simple node $w$ with $\CELabelI(w) = a$ can be treated as an edge from the named individual $a$ (an inner node of the graph representing $v$) to the unnamed individual corresponding to $w$, and that edge is via the roles from $\CELabelR(w)$.

\item $G$ consists of two layers: the layer of complex nodes and the layer of simple nodes. The former layer consists of the local graph of the root $\nu$ together with a number of complex states and edges coming to them. The edges from the layer of complex nodes to the layer of simple nodes are exactly the edges outgoing from those complex states. There are no edges from the layer of simple nodes to the layer of complex nodes. 
\koniec
\end{itemize}
\end{remark}

We apply global state caching: if $v_1$ and $v_2$ are different states then $\Label(v_1) \neq \Label(v_2)$. 
If $v$ is a non-state such that $\AfterTrans(v)$ holds then we also apply global caching for the local graph of $v$: if $w_1$ and $w_2$ are different nodes of the local graph of $v$ then $\Label(w_1) \neq \Label(w_2)$. 
Creation of a new node or a new edge is done by Procedures $\ConToSucc$ (connect to a successor) or $\NewSucc$ (new successor) given on page~\pageref{proc: NewSucc}. These procedures create a connection from a node $v$ given as the first parameter to a node with attributes $\Type$, $\SType$, $\Label$, $\RFormulas$, $\CELabel$ specified by the remaining parameters. The difference between these procedures is that $\NewSucc$ always creates a new node and a new connection, while $\ConToSucc$ first checks whether there exists a node that can be used as a proxy for the successor to be created. $\NewSucc$ is called outside $\ConToSucc$ only when $v$ (the first parameter) is a state. 


\begin{figure*}[ht!]
\begin{function}[H]
\caption{NewSucc($v, type, sType, label, rFmls, ceLabel$)\label{proc: NewSucc}}
\GlobalData{a rooted graph $\tuple{V,E,\nu}$.}
\Purpose{create a new successor for $v$.}
\LinesNumberedHidden
create a new node $w$,\ \ 
$V := V \cup \{w\}$,\ \ 
\lIf{$v \neq \Null$}{$E := E \cup \{\tuple{v,w}\}$}\;

$\Type(w) := type$,\ 
$\SType(w) := sType$,\ 
$\Status(w) := \Unexpanded$\; 
$\Label(w) := label$,\ 
$\RFormulas(w) := rFmls$\; 

\uIf{$type = \NonState$}{
   \lIf{$v = \Null$ or $\Type(v) = \State$}{$\StatePred(w) := v$,\ \ 
	$\AfterTransPred(w) := w$\\
   }
   \lElse{$\StatePred(w) := \StatePred(v)$,\ \ 
	$\AfterTransPred(w) := \AfterTransPred(v)$}\;
   \lIf{$\Type(v) = \State$}{$\CELabel(w) := ceLabel$}\;
}
\lElse{$\FmlsRC(w) := \emptyset$,\ \ 
   $\ILConstraints(w) := \emptyset$}\;

\Return{$w$}
\end{function}

\medskip

\begin{function}[H]
\caption{FindProxy($type, sType, v_1, label$)\label{proc: FindProxy}}
\GlobalData{a rooted graph $\tuple{V,E,\nu}$.}
\LinesNumberedHidden

\lIf{$type = \State$}{$W := V$}
\lElse{$W := $ the nodes of the local graph of $v_1$}\;

\lIf{there exists $w \in W$ such that $\Type(w) = type$ and $\SType(w) = sType$ and $\Label(w) = label$}{\Return $w$}
\lElse{\Return $\Null$}
\end{function}

\medskip

\begin{function}[H]
\caption{ConToSucc($v, type, sType, label, rFmls, ceLabel$)\label{proc: ConToSucc}}
\GlobalData{a rooted graph $\tuple{V,E,\nu}$.}
\Purpose{connect a non-state $v$ to a successor, which is created if necessary.}

\lIf{$type = \State$}{$v_1 := \Null$}
\lElse{$v_1 := \AfterTransPred(v)$}\;

$w := \FindProxy(type, sType, v_1, label)$\;
\lIf{$w \neq \Null$}{$E := E \cup \{\tuple{v,w}\}$,\ \ $\RFormulas(w) := \RFormulas(w) \cup rFmls$}\\
\lElse{$w := \NewSucc(v, type, sType, label, rFmls, ceLabel)$}\;

\Return{$w$}
\end{function}
\end{figure*}


From now on, let $\tuple{\mR,\mT,\mA}$ be a knowledge base in NNF of the logic \SHIQ, with $\mA \neq \emptyset$.$\,$\footnote{If $\mA$ is empty, we can add $a\!:\!\top$ to it, where $a$ is a special individual.} In this section we present a tableau calculus \CSHIQ for checking satisfiability of $\tuple{\mR,\mT,\mA}$.
A~\CSHIQ-tableau for $\tuple{\mR,\mT,\mA}$ is a rooted graph $G = \tuple{V,E,\nu}$ constructed as follows:

\subsubsection*{Initialization:} 

At the beginning, $V = \{\nu\}$, $E = \emptyset$, and $\nu$ is the root with $\Type(\nu) = \NonState$, $\SType(\nu) = \Complex$, $\Status(\nu) = \Unexpanded$, $\Label(\nu) = \mA \cup \{(a\!:\!C) \mid C \in \mT$ and $a$ is an individual occurring in $\mA\}$, 
$\StatePred(\nu) = \Null$, 
$\RFormulas(\nu) = \emptyset$. 

\subsubsection*{Rules' Priorities and Expansion Strategies:} 

The graph is then expanded by the following rules, which will be specified shortly:

  \begin{tabular}{cl}
  \UPS & rules for updating statuses of nodes,\\
  \US & a unary static expansion rule,\\
  \KCC & rules for keeping converse compatibility,\\ 
  \NUS & a non-unary static expansion rule,\\ 
  \FS & forming-state rules,\\ 
  \TP & a transitional partial-expansion rule,\\ 
  \TF & a transitional full-expansion rule.
  \end{tabular}

Each of the rules is parametrized by a node $v$. We say that a rule is {\em applicable} to $v$ if it can be applied to $v$ to make changes to the graph. 
The rules \UPS, \US, \KCC (in the first three items of the above list) have highest priorities, and are ordered decreasingly by priority. If none of them is applicable to any node, then choose a node $v$ with status $\Unexpanded$ or $\PExpanded$, choose the first rule applicable to $v$ among the rules in the last four items of the above list, and apply it to $v$. Any strategy can be used for choosing $v$, but it is worth choosing $v$ for expansion only when $v$ could affect the status of the root $\nu$ of the graph, i.e., only when there may exist a path from $\nu$ to $v$ without any node of status $\Incomplete$, $\Unsat$ or $\Sat$. 

Note that the priorities of the rules are specified by the order in the above list, but the rules \UPS, \US, \KCC are checked globally (technically, they are triggered immediately when possible), while the remaining rules are checked for a chosen node.

The construction of the graph ends when the root $\nu$ receives status $\Unsat$ or $\Sat$ or when no more changes that may affect the status of $\nu$ can be made\footnote{That is, ignoring nodes that are unreachable from $\nu$ via a path without nodes with status $\Incomplete$, $\Unsat$ or $\Sat$, no more changes can be made to the graph.}


\begin{remark}\label{remark: main intuition}
To give a deeper intuition behind the structure of a \CSHIQ-tableau $G$ constructed for a knowledge base $\KB = \tuple{\mR,\mT,\mA}$, assume that the knowledge base is satisfiable and let us briefly describe how a model $\mI$ for $\KB$ can be constructed from $G$. A precise description will be given in the proof of completeness of the calculus. As we will see, the root $\nu$ must have status different from $\Unsat$ and there must exist a complex state $v_k$ with status different from $\Unsat$. $\Label(v_k)$ is an ABox that forms the base for the constructed model $\mI$. Each named individual $a$ occurring in that ABox is used as an element of the domain of $\mI$; role assertions in that ABox specify edges between those elements in $\mI$; concept assertions in that ABox specify whether a named individual $a$ is an instance of a given atomic concept in $\mI$. The rest of $\mI$ consists of disjoint trees rooted at those named elements, which may be infinite. Each element of the domain of $\mI$ corresponds either to one of those named individuals or to a simple state of $G$ with status different from $\Unsat$. It is possible that only a number of simple states of $G$ are used for constructing $\mI$, and each simple state of $G$ may correspond to many elements of the domain of~$\mI$ (due to global state caching). 

Let $y$ be an element of the domain of $\mI$. How successors of $y$ in $\mI$ can be constructed? 

If $y$ corresponds to a named individual $a$ then let $W$ be the set of successors $w$ of $v_k$ such that: either $\CELabelT(w) = \TUS$ and we set $n_w = 1$; or $\CELabelT(w) = \CQF$, $\CELabelI(w) = a$ and the value of $x_w$ in a fixed solution for $\ILConstraints(v_k)$, denoted by $n_w$, is greater than~0. 

If $y$ corresponds to a simple state $u$ of $G$ then let $W$ be the set of successors $w$ of $u$ such that: either $\CELabelT(w) = \TUS$ and we set $n_w = 1$; or $\CELabelT(w) = \CQF$ and the value of $x_w$ in a fixed solution for $\ILConstraints(u)$, denoted by $n_w$, is greater than~0. 

For each $w \in W$, there must exist a path in $G$ starting from (the non-state) $w$, going through some or zero other non-states and ending at a state $w'$ with status different from $\Unsat$. For each $w \in W$, we create $n_w$ successors for $y$ (in $\mI$) that correspond to $w'$, using edges labeled by the roles from $\CELabelR(w)$. Roughly speaking, all nodes in the path from $w$ to $w'$ are stitched together to make a successor $z$ for $y$, which is then cloned $n_w$ times. The set of atomic concepts $A$ such a successor $z$ must be an instance of in $\mI$ (i.e., $z \in A^\mI$) consists of the atomic concepts in $\Label(w')$. 

Notice that non-states are like ``or''-nodes, while states are more sophisticated than ``and''-nodes. The transitional partial-expansion rule deals with non-numeric roles, while the transitional full-expansion rule deals with numeric roles.
\koniec
\end{remark}


\subsection{Illustrative Examples}

Before specifying the tableau rules in detail, we present simple examples to illustrate some ideas (but not all aspects) of our method. Despite that these examples refer to the tableau rules, we choose this place for presenting them because the examples are quite intuitive and the reader can catch the ideas of our method without knowing the detailed rules. He or she can always consult the rules in the next subsection. 

\newcommand{\myhline}{\\[0.4ex] \hline \\[-1.6ex]}

\begin{figure}
\begin{center}
\begin{tabular}{c}
\begin{scriptsize}
\begin{tabular}{c@{\extracolsep{10em}}c}
(a) & (b) \\
\\
\xymatrix{
*+[F]{\begin{tabular}{c}
	$\nu$
	\myhline
	$a:\!:\!\top$, $a\!:\!\E r.(A \mand \V s.\lnot A)$
      \end{tabular}}
\ar@{->}[d] 
\\
*+[F=]{\begin{tabular}{c}
	$v_1$
	\myhline
	$a:\!:\!\top$, $a\!:\!\E r.(A \mand \V s.\lnot A)$
      \end{tabular}}
\ar@{->}[d] 
\\
*+[F]{\begin{tabular}{c}
	$v_2$
	\myhline
	$A \mand \V s.\lnot A$,\\
	$\E r.(A \mand \V s.\lnot A)$
      \end{tabular}}
\ar@{->}[d]
\\
*+[F]{\begin{tabular}{c}
	$v_3$
	\myhline
	$A$, $\V s.\lnot A$,\\
	$\E r.(A \mand \V s.\lnot A)$
      \end{tabular}}
\ar@{->}[d]
\\
*+[F]{\begin{tabular}{c}
	$v_4$
	\myhline
	$A$, $\V s.\lnot A$,\\ 
	$\V r.\lnot A$, $\V r^-.\lnot A$,\\
	$\E r.(A \mand \V s.\lnot A)$
      \end{tabular}}
}
&
\xymatrix{
*+[F]{\begin{tabular}{c}
	$\nu$
	\myhline
	$a:\!:\!\top$, $a\!:\!\E r.(A \mand \V s.\lnot A)$
      \end{tabular}}
\ar@{->}[d] 
\\
*+[F=]{\begin{tabular}{c}
	$v_1$
	\myhline
	$a:\!:\!\top$, $a\!:\!\E r.(A \mand \V s.\lnot A)$\\
	$\Incomplete$\\
	$\FmlsRC$ : $a\!:\!\lnot A$, $a\!:\!\V s.\lnot A$
      \end{tabular}}
\ar@{->}[d] 
\\
*+[F]{\begin{tabular}{c}
	$v_2$
	\myhline
	$A \mand \V s.\lnot A$,\\
	$\E r.(A \mand \V s.\lnot A)$
      \end{tabular}}
\ar@{->}[d]
\\
*+[F]{\begin{tabular}{c}
	$v_3$
	\myhline
	$A$, $\V s.\lnot A$,\\
	$\E r.(A \mand \V s.\lnot A)$
      \end{tabular}}
\ar@{->}[d]
\\
*+[F]{\begin{tabular}{c}
	$v_4$
	\myhline
	$A$, $\V s.\lnot A$,\\
	$\V r.\lnot A$, $\V r^-.\lnot A$,\\
	$\E r.(A \mand \V s.\lnot A)$
      \end{tabular}}
}
\end{tabular}
\end{scriptsize}
\end{tabular}
\end{center}
\caption{\label{fig: example0}An illustration for Example~\ref{ex:andorgraph}: part I. The graph (a) is the graph constructed until checking converse compatibility for $v_1$. In each node, we display the formulas of the node's label. The root $\nu$ is expanded by the forming-state rule \FSb. The complex state $v_1$ is expanded by the transitional partial-expansion rule, with $\CELabelT(v_2) = \TUS$, $\CELabelR(v_2) = \{r,s,s^-\}$ ($s$ and $s^-$ are included because $r \sqsubseteq_\mR s$ and $r \sqsubseteq_\mR s^-$) and $\CELabelI(v_2) = a$. The simple non-states $v_2$ and $v_3$ are expanded by the unary static expansion rule (the concepts $\V r.\lnot A$ and $\V r^-.\lnot A$ are added to $\Label(v_4)$ because $\V s.\lnot A \in \Label(v_3)$, $r \sqsubseteq_\mR s$ and $r^- \sqsubseteq_\mR s$). The node $v_1$ is the only state. We have, for example, $\StatePred(v_4) = v_1$ and $\AfterTransPred(v_4) = v_2$. Checking converse compatibility for $v_1$ using $v_4$ (i.e., using the facts that $\{\V s.\lnot A, \V r^-.\lnot A\} \subset \Label(v_4)$, $\StatePred(v_4) = v_1$, $\AfterTransPred(v_4) = v_2$, $\CELabelT(v_2) = \TUS$, $\CELabelR(v_2) = \{r,s,s^-\}$, $\CELabelI(v_2) = a$, $r^- \sqsubseteq_\mR s$ and $\Trans{R}$ holds), $\Status(v_1)$ is set to $\Incomplete$ and $\FmlsRC(v_1)$ is set to $\{a\!:\!\lnot A, a\!:\!\V s.\lnot A\}$. This results in the graph~(b). The construction is then continued by re-expanding the node $\nu$. See Figure~\ref{fig: example0-II} for the continuation.} 
\end{figure}
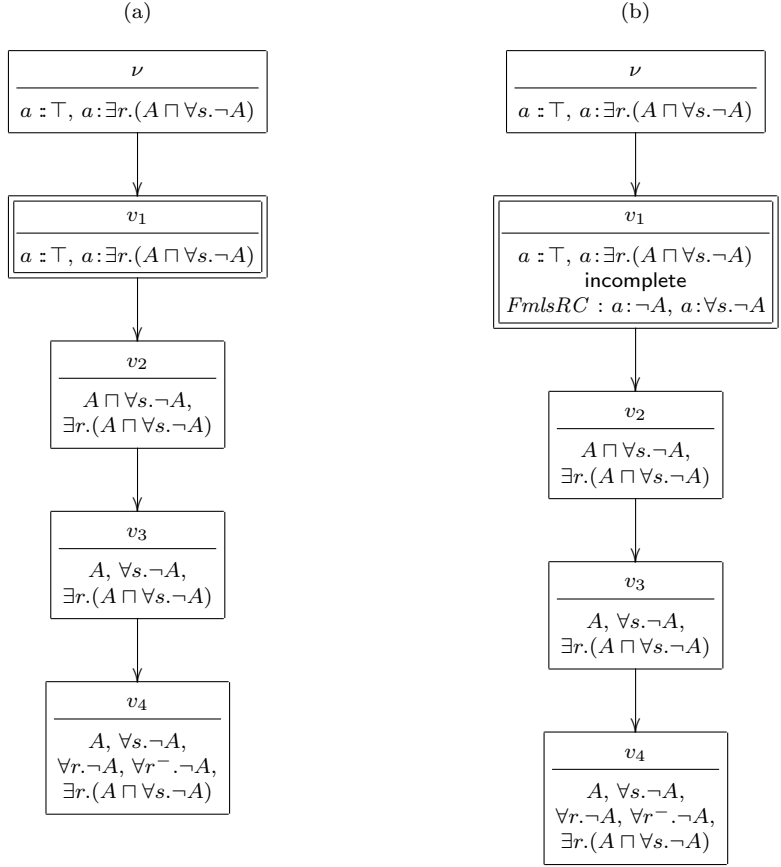

\begin{figure}
\begin{center}
\begin{tabular}{c}
\begin{scriptsize}
\xymatrix{
*+[F]{\begin{tabular}{c}
	$\nu$
	\myhline
	$a:\!:\!\top$, $a\!:\!\E r.(A \mand \V s.\lnot A)$
      \end{tabular}}
\ar@{->}[dr] 
\\
*+[F=]{\begin{tabular}{c}
	$v_1$
	\myhline
	$a:\!:\!\top$, $a\!:\!\E r.(A \mand \V s.\lnot A)$\\
	$\Incomplete$\\
	$\FmlsRC$ : $a\!:\!\lnot A$, $a\!:\!\V s.\lnot A$
      \end{tabular}}
\ar@{->}[d] 
&
*+[F]{\begin{tabular}{c}
	$v_5$
	\myhline
	$a:\!:\!\top$, $a\!:\!\E r.(A \mand \V s.\lnot A)$,\\ 
	$a\!:\!\lnot A$, $a\!:\!\V s.\lnot A$
      \end{tabular}}
\ar@{->}[d] 
\\
*+[F]{\begin{tabular}{c}
	$v_2$
	\myhline
	$A \mand \V s.\lnot A$,\\
	$\E r.(A \mand \V s.\lnot A)$
      \end{tabular}}
\ar@{->}[d]
&
*+[F]{\begin{tabular}{c}
	$v_6$
	\myhline
	$a:\!:\!\top$, $a\!:\!\E r.(A \mand \V s.\lnot A)$,\\ 
	$a\!:\!\lnot A$, $a\!:\!\V s.\lnot A$,\\
	$a\!:\!\V r.\lnot A$, $a\!:\!\V r^-.\lnot A$
      \end{tabular}}
\ar@{->}[d]
\\
*+[F]{\begin{tabular}{c}
	$v_3$
	\myhline
	$A$, $\V s.\lnot A$,\\
	$\E r.(A \mand \V s.\lnot A)$
      \end{tabular}}
\ar@{->}[d]
&
*+[F=]{\begin{tabular}{c}
	$v_7$
	\myhline
	$a:\!:\!\top$, $a\!:\!\E r.(A \mand \V s.\lnot A)$,\\ 
	$a\!:\!\lnot A$, $a\!:\!\V s.\lnot A$,\\
	$a\!:\!\V r.\lnot A$, $a\!:\!\V r^-.\lnot A$
      \end{tabular}}
\ar@{->}[d]
\\
*+[F]{\begin{tabular}{c}
	$v_4$
	\myhline
	$A$, $\V s.\lnot A$,\\
	$\V r.\lnot A$, $\V r^-.\lnot A$,\\
	$\E r.(A \mand \V s.\lnot A)$
      \end{tabular}}
&
*+[F]{\begin{tabular}{c}
	$v_8$
	\myhline
	$A \mand \V s.\lnot A$,\\ 
	$\lnot A$, $\V s.\lnot A$,\\
	$\E r.(A \mand \V s.\lnot A)$
      \end{tabular}}
\ar@{->}[d]
\\
&
*+[F]{\begin{tabular}{c}
	$v_9$
	\myhline
	$A$, $\V s.\lnot A$, $\lnot A$,\\
	$\E r.(A \mand \V s.\lnot A)$\\
	$\Unsat$
      \end{tabular}}
}
\end{scriptsize}
\end{tabular}
\end{center}
\caption{\label{fig: example0-II}An illustration for Example~\ref{ex:andorgraph}: part II. This is a \CSHIQ-tableau for $\tuple{\mR,\mT,\mA}$. As in the part~I, in each node we display the formulas of the node's label. The root $\nu$ is re-expanded by deleting the edge $\tuple{\nu,v_1}$ and connecting $\nu$ to a new complex non-state $v_5$. The node $v_5$ is expanded using the unary static expansion rule (the assertions $a\!:\!\V r.\lnot A$ and $a\!:\!\V r^-.\lnot A$ are added to $\Label(v_6)$ because $a\!:\!\V s.\lnot A \in \Label(v_5)$, $r \sqsubseteq_\mR s$ and \mbox{$r^- \sqsubseteq_\mR s$}). The complex non-state $v_6$ is expanded using the forming-state rule \FSb. The complex state $v_7$ is expanded using the transitional partial-expansion rule, with $\CELabelT(v_8) = \TUS$, $\CELabelR(v_8) = \{r,s,s^-\}$ ($s$ and $s^-$ are included because $r \sqsubseteq_\mR s$ and $r \sqsubseteq_\mR s^-$) and $\CELabelI(v_8) = a$. The simple non-state $v_8$ is expanded using the unary static expansion rule. Since $\{A,\lnot A\} \subset \Label(v_9)$, the simple non-state $v_9$ receives status $\Unsat$. After that the nodes $v_8, \ldots, v_5$ and $\nu$ receive status $\Unsat$ in subsequent steps. The nodes $v_1$ and $v_7$ are the only states. 
} 
\end{figure}

\begin{example}\label{ex:andorgraph}
This example is based on an example in the long version of our paper~\cite{SHI-ICCCI} on \SHI. 
Let 
\begin{eqnarray*}
\mR & = & \{r \sqsubseteq s,\ \ r^- \sqsubseteq s,\ \ s \circ s \sqsubseteq s\}\\
\mT & = & \{\E r.(A \mand \V s.\lnot A)\}\\
\mA & = & \{ a\!:\!\top \}.
\end{eqnarray*}
In Figures~\ref{fig: example0} and \ref{fig: example0-II} we illustrate the construction of a \CSHIQ-tableau for the knowledge base $\tuple{\mR,\mT,\mA}$. 
At the end the root $\nu$ receives status $\Unsat$. By Theorem~\ref{theorem: s-c} given later in this paper, $\tuple{\mR,\mT,\mA}$ is unsatisfiable. As a consequence, $\tuple{\mR,\mT,\emptyset}$ is also unsatisfiable.  
\koniec
\end{example}


\begin{example}\label{example2}
Let us construct a \CSHIQ-tableau for $\tuple{\mR,\mT,\mA}$, where $\mR = \emptyset$, $\mT = \emptyset$ and 
\[ \mA = \{a\!:\!\V r.A_1,\ a\!:\,\leq\!3\,r.A_1,\ a\!:\!\E r.A_2,\ a\!:\,\geq\!2\,r.A_3,\ r(a,b),\ b\!:\!(\lnot A_1 \mor \lnot A_2),\ b\!:\!\lnot A_3\} \]

\begin{figure}[t]
\begin{center}
\includegraphics{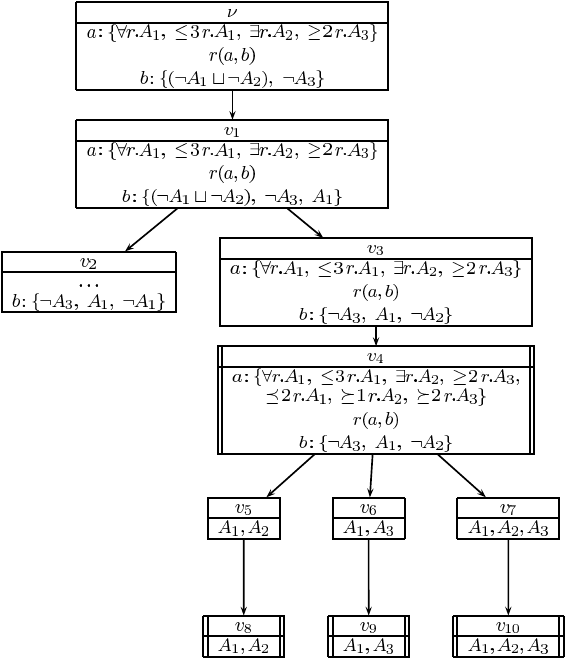}
\end{center}
\caption{An illustration of the tableau described in Example~\ref{example2}. 
The marked nodes $v_4$, $v_8$, $v_9$, $v_{10}$ are states. 
\label{fig-2}}
\end{figure} 

An illustration of the tableau is given in Figure~\ref{fig-2}.

At the beginning, the graph has only the root $\nu$ which is a complex non-state with \mbox{$\Label(\nu) = \mA$}. 

Since $\{a\!:\!\V r.A_1, r(a,b)\} \subset \Label(\nu)$, applying the unary static expansion rule to $\nu$, we connect it to a new complex non-state $v_1$ with $\Label(v_1) = \Label(\nu) \cup \{b\!:\!A_1\}$. 

Since $b\!:\!(\lnot A_1 \mor \lnot A_2) \in \Label(v_1)$, applying the non-unary static expansion rule to $v_1$, we connect it to two new complex non-states $v_2$ and $v_3$ with 
\begin{eqnarray*}
\Label(v_2) & = & \Label(v_1) - \{b\!:\!(\lnot A_1 \mor \lnot A_2)\} \cup \{b\!:\!\lnot A_1\}\\
\Label(v_3) & = & \Label(v_1) - \{b\!:\!(\lnot A_1 \mor \lnot A_2)\} \cup \{b\!:\!\lnot A_2\}. 
\end{eqnarray*}

Since both $b\!:\!A_1$ and $b\!:\!\lnot A_1$ belong to $\Label(v_2)$, the node $v_2$ receives status $\Unsat$. 

Applying the forming-state rule \FSb to $v_3$, we connect it to a new complex state $v_4$ with 
\[ \Label(v_4) = \Label(v_3) \cup \{a\!:\,\preceq\!2\,r.A_1,\ a\!:\,\succeq\!1\,r.A_2,\ a\!:\,\succeq\!2\,r.A_3\}. \]
The assertion $a\!:\,\preceq\!2\,r.A_1 \in \Label(v_4)$ is due to $a\!:\,\leq\!3\,r.A_1 \in \Label(v_3)$ and the fact that $\{r(a,b)$, $b\!:\!A_1\} \subset \Label(v_3)$. 
The assertion $a\!:\,\succeq\!1\,r.A_2 \in \Label(v_4)$ is due to $a\!:\!\E r.A_2 \in \Label(v_3)$ and the fact that $b\!:\!\lnot A_2 \in \Label(v_3)$. Similarly, the assertion $a\!:\,\succeq\!2\,r.A_3 \in \Label(v_4)$ is due to $a\!:\,\geq\!2\,r.A_3 \in \Label(v_3)$ and the fact that $b\!:\!\lnot A_3 \in \Label(v_3)$. 
When realizing the requirements for $a$ at $v_4$, we will not have to pay attention to the relationship between $a$ and $b$. 

Applying the transitional partial-expansion rule to $v_4$, we change its status to $\PExpanded$. After that, applying the transitional full-expansion rule to $v_4$, we connect it to new simple non-states $v_5$, $v_6$ and $v_7$ with $\CELabelT$ equal to $\CQF$, $\CELabelR$ equal to $\{r\}$, $\CELabelI$ equal to $a$, $\Label(v_5) = \{A_1,A_2\}$, $\Label(v_6) = \{A_1,A_3\}$ and $\Label(v_7) = \{A_1,A_2,A_3\}$. The creation of $v_5$ is caused by $a\!:\,\succeq\!1\,r.A_2 \in \Label(v_4)$, while the creation of $v_6$ is caused by \mbox{$a\!:\,\succeq\!2\,r.A_3 \in \Label(v_4)$}. The node $v_7$ results from merging $v_5$ and $v_6$. Furthermore, $\ILConstraints(v_4)$ consists of $x_{v_i} \geq 0$, for $5 \leq i \leq 7$, and 
\begin{eqnarray*}
x_{v_5} + x_{v_6} + x_{v_7} & \leq & 2\\
x_{v_5} + x_{v_7} & \geq & 1\\
x_{v_6} + x_{v_7} & \geq & 2.
\end{eqnarray*}

The set $\ILConstraints(v_4)$ is feasible, e.g., with $x_{v_5} = 0$ and $x_{v_6} = x_{v_7} = 1$. 

Applying the forming-state rule \FSa to $v_5$ (resp.\ $v_6$, $v_7$), we connect it to a new simple state $v_8$ (resp.\ $v_9$, $v_{10}$) with the same label.

Expanding the nodes $v_8$, $v_9$ and $v_{10}$, their statuses change to $\PExpanded$ and then to $\FExpanded$. The graph cannot be modified anymore and becomes a \CSHIQ-tableau for $\tuple{\mR,\mT,\mA}$, with $\Status(\nu) \neq \Unsat$. 
By Theorem~\ref{theorem: s-c} (given later in this paper), the considered knowledge base $\tuple{\mR,\mT,\mA}$ is satisfiable. Using the solution $x_{v_5} = 0$, $x_{v_6} = x_{v_7} = 1$ for $\ILConstraints(v_4)$, we can extract from the graph a model $\mI$ for $\tuple{\mR,\mT,\mA}$ with $\Delta^\mI = \{a,b,v_9,v_{10}\}$, $a^\mI = a$, $b^\mI = b$, $A_1^\mI = \{b,v_9,v_{10}\}$, $A_2^\mI = \{v_{10}\}$, $A_3^\mI = \{v_9,v_{10}\}$ and $r^\mI = \{\tuple{a,b},\tuple{a,v_9},\tuple{a,v_{10}}\}$.
\koniec
\end{example}


\begin{example}\label{example1}
Let us construct a \CSHIQ-tableau for $\tuple{\mR,\mT,\mA}$, where $\mR = \emptyset$, $\mT = \emptyset$ and 
\[ \mA = \{
	a\!:\!\E r.A,\; 
	a\!:\!\V r.(\lnot A \mor \lnot B),\; 
	a\!:\,\geq\!1000\,r.B,\; 
	a\!:\,\leq\!1000\,r.(A \mor B)\}. 
\]

An illustration of the tableau is given in Figure~\ref{fig-1}.

\begin{figure}[t]
\begin{center}
\includegraphics{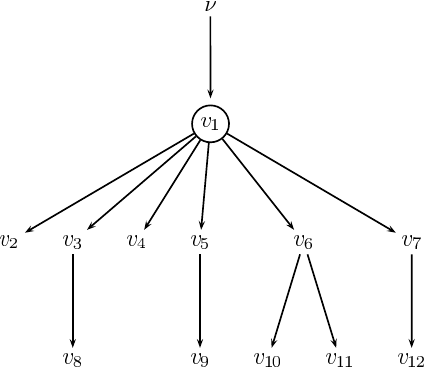}
\end{center}
\caption{An illustration of the tableau described in~Example~\ref{example1}.}
\label{fig-1}
\end{figure}

At the beginning, the graph has only the root $\nu$ which is a complex non-state with $\Label(\nu) = \mA$. Applying the forming-state rule \FSb to $\nu$, we connect it to a new complex state $v_1$ with 
\[ \Label(v_1) = \Label(\nu) \cup \{
	a\!:\,\succeq\!1\,r.A,\;
	a\!:\,\succeq\!1000\,r.B,\; 
	a\!:\,\preceq\!1000\,r.(A \mor B)\}. 
\]

Applying the transitional partial-expansion rule to $v_1$, we change its status to $\PExpanded$. After that, applying the transitional full-expansion rule to $v_1$, connect it to new simple non-states $v_2, \ldots, v_7$ with $\CELabelT$ equal to $\CQF$, $\CELabelR$ equal to $\{r\}$, $\CELabelI$ equal to $a$, and 
\begin{eqnarray*}
\Label(v_2) & = & \{A,\; \lnot A \mor \lnot B,\; A \mor B\} \\
\Label(v_3) & = & \{A,\; \lnot A \mor \lnot B,\; \lnot A \mand \lnot B\} \\
\Label(v_4) & = & \{B,\; \lnot A \mor \lnot B,\; A \mor B\} \\
\Label(v_5) & = & \{B,\; \lnot A \mor \lnot B,\; \lnot A \mand \lnot B\} \\
\Label(v_6) & = & \{A,\; B,\; \lnot A \mor \lnot B,\; A \mor B\} \\
\Label(v_7) & = & \{A,\; B,\; \lnot A \mor \lnot B,\; \lnot A \mand \lnot B\}.
\end{eqnarray*}
Note that $\lnot A \mand \lnot B$ is the NNF of $A \mor B$. The nodes $v_2$ and $v_3$ are created due to $a\!:\,\succeq\!1\,r.A \in \Label(v_1)$. The nodes $v_4$ and $v_5$ are created due to $a\!:\,\succeq\!1000\,r.B \in \Label(v_1)$. The node $v_6$ results from merging $v_2$ and $v_4$. The node $v_7$ results from merging $v_3$ and $v_5$. 
Furthermore, $\ILConstraints(v_1)$ consists of $x_{v_i} \geq 0$, for $2 \leq i \leq 7$, and 
\begin{eqnarray*}
x_{v_2} + x_{v_3} + x_{v_6} + x_{v_7} & \geq & 1\\
x_{v_4} + x_{v_5} + x_{v_6} + x_{v_7} & \geq & 1000\\
x_{v_2} + x_{v_4} + x_{v_6} & \leq & 1000.
\end{eqnarray*}

Applying the unary static expansion rule to $v_3$, we connect it to a new simple non-state $v_8$ with $\Label(v_8) = \{A,\lnot A, \ldots\}$. The status of $v_8$ is then updated to $\Unsat$ and propagated back to make the status of $v_3$ also become $\Unsat$, which causes addition of the constraint $x_{v_3} = 0$ to the set $\ILConstraints(v_1)$.

Similarly, applying the unary static expansion rule to $v_5$, we connect it to a new simple non-state $v_9$ with $\Label(v_9) = \{B,\lnot B, \ldots\}$. The status of $v_9$ is then updated to $\Unsat$ and propagated back to make the status of $v_5$ also become $\Unsat$, which causes addition of the constraint $x_{v_5} = 0$ to the set $\ILConstraints(v_1)$.

Applying the non-unary static expansion rule to $v_6$, we connect it to new simple non-states $v_{10}$ and $v_{11}$ with $\Label(v_{10}) = \{A,\lnot A, \ldots\}$ and $\Label(v_{11}) = \{B,\lnot B, \ldots\}$. The statuses of $v_{10}$ and $v_{11}$ are then updated to $\Unsat$ and propagated back to make the status of $v_6$ also become $\Unsat$, which causes addition of the constraint $x_{v_6} = 0$ to the set $\ILConstraints(v_1)$.

Applying the unary static expansion rule to $v_7$, we connect it to a new simple non-state $v_{12}$ with $\Label(v_{12}) = \{A,B,\lnot A,\lnot B, \ldots\}$. The status of $v_{12}$ is then updated to $\Unsat$ and propagated back to make the status of $v_7$ also become $\Unsat$, which causes addition of the constraint $x_{v_7} = 0$ to the set $\ILConstraints(v_1)$.

With $x_{v_3} = x_{v_5} = x_{v_6} = x_{v_7} = 0$, $\ILConstraints(v_1)$ becomes infeasible, and $\Status(v_1)$ becomes $\Unsat$. Consequently, $\nu$ receives status $\Unsat$. 
By Theorem~\ref{theorem: s-c} (given later in this paper), the considered knowledge base $\tuple{\mR,\mT,\mA}$ is unsatisfiable. 
\koniec
\end{example}


\subsection{The Tableau Rules}

In this section we formally specify the tableau rules of our calculus \CSHIQ. 
We also give explanations for them. They are informal and should be understood in the context of the described rule. 

\sssection{The Rules for Updating Statuses of Nodes:} 

\begin{description}
\item[\UPSa] The first rule is as follows:
  \begin{enumerate}
  \item if there exists $\alpha\!:\!\bot \in \Label(v)$ or $\{\varphi,\ovl{\varphi}\} \subseteq \FullLabel(v)$ then $\Status(v) := \Unsat$
  \item else if there exists $a \not\doteq a \in \Label(v)$ then $\Status(v) := \Unsat$
  \item else if $\Type(v) = \NonState$, $\SType(v) = \Complex$, $a\!:\!(\leq\!n R.C) \in \Label(v)$ and there are $b_0,\ldots,b_n \in \IN$ such that, for all $0 \leq i,j \leq n$ with $i \neq j$, we have that $\{R(a,b_i), b_i\!:\!C\} \subseteq \FullLabel(v)$ and $\{b_i \not\doteq b_j, b_j \not\doteq b_i\} \cap \Label(v) \neq \emptyset$ then $\Status(v) := \Unsat$
  \item else if $\Type(v) = \State$, $\Status(v) = \FExpanded$ and $\ILConstraints(v)$ is infeasible then $\Status(v) := \Unsat$
  \item else if $\Type(v) = \State$, $\Status(v) = \FExpanded$ and $v$ has no successors then $\Status(v) := \Sat$.
  \end{enumerate}

  \begin{explanation}
  As stated earlier, informally, $\Unsat$ means ``unsatisfiable'' and $\Sat$ means ``satisfiable'' in a certain sense. The above rule is thus intuitive. For a formal characterization of statuses, we refer the reader to Lemma~\ref{lemma: SHQWD} (given on page~\pageref{lemma: SHQWD}).
\koniec
  \end{explanation}

\item[\UPSb] The second rule states that, if $v$ is a predecessor of a node $w$ then, whenever the status of $w$ changes to $\Unsat$ or $\Sat$, the status of $v$ should be updated (as soon as possible by using a priority queue of tasks). The update is done as follows: 
  \begin{enumerate}
  \item if $\Type(v) = \NonState$ then
     \begin{enumerate}
     \item\label{item: IUQKD} if some successor of $v$ received status $\Sat$ then $\Status(v) := \Sat$
     \item else if all successors of $v$ have status $\Unsat$ then
	\begin{enumerate}
	\item\label{item: OIWGS} $\Status(v) := \Unsat$
	\item\label{item: OFDSH} if $\AfterTrans(v)$ holds, $\CELabelT(v) = \CQF$, $\StatePred(v) = u$ then add the constraint $x_v = 0$ to the set $\ILConstraints(u)$
	\end{enumerate}
     \end{enumerate}
  \item else 
     \begin{enumerate}
     \item\label{item: YUWQP} if some successor $w$ of $v$ with $\CELabelT(w) = \TUS$ received status $\Unsat$ then $\Status(v) := \Unsat$
     \item\label{item: HRUIA} else if some successor $w$ of $v$ with $\CELabelT(w) = \CQF$ received status $\Unsat$ and $\ILConstraints(v)$ is infeasible then $\Status(v) := \Unsat$
     \item\label{item: HFDDE} else if 
	\begin{itemize}
	\item $\Status(v) = \FExpanded$, 
	\item all successors $w$ of $v$ with $\CELabelT(w) = \TUS$ have status $\Sat$, and 
	\item $\ILConstraints(v) \cup \{x_w = 0 \mid$ $\tuple{v,w} \in E$, $\CELabelT(w) = \CQF$ and $\Status(w) \neq \Sat\}$ is feasible 
	\end{itemize}
     then $\Status(v) := \Sat$.
     \end{enumerate}
  \end{enumerate}

  \begin{explanation}\ 
  \begin{itemize}
  \item A non-state is like an ``or''-node, whose status is the disjunction of the statuses of its successors, treating $\Sat$ as $\True$ and $\Unsat$ as $\False$. This explains the items \ref{item: IUQKD} and \ref{item: OIWGS}. 
  \item A state $v$ is more sophisticated than an ``and''-node. Its status is different from $\Unsat$ iff the following conditions hold:
    \begin{itemize}
    \item all of its successors with $\CELabelT$ equal to $\TUS$ have status different from $\Unsat$,
    \item $\ILConstraints(v) \cup \{x_w = 0 \mid$ $\tuple{v,w} \in E$, $\CELabelT(w) = \CQF$ and $\Status(w) = \Unsat\}$ is feasible.
    \end{itemize}
  The first condition explains the item~\ref{item: YUWQP}. The second condition explains the item~\ref{item: HRUIA} because, whenever a successor $w$ of $v$ with $\CELabelT$ equal to $\CQF$ receives status $\Unsat$, the constraint $x_w = 0$ is added to $\ILConstraints(v)$ (see the item~\ref{item: OFDSH}). The item~\ref{item: HFDDE} is justifiable since the premises of the rule expressed by that item are stronger than the conditions listed above.
\koniec
  \end{itemize}
  \end{explanation}
\end{description}

\sssection{The Unary Static Expansion Rule:}
  \begin{description}
  \item[\US] If $\Type(v) = \NonState$ and $\Status(v) = \Unexpanded$ then 
    \begin{enumerate}
     \item let $X = \{(\alpha\!:\!C) \in \Label(v) \mid C$ is of the form $D \mand D'$ or $\geq\!0\,R.D$ or $\leq\!0\,R.D\}$ 
     \item let $label = \Label(v) \cup \{(\alpha\!:\!D), (\alpha\!:\!D') \mid \alpha\!:\!(D \mand D') \in \Label(v)\}$ \\
	\mbox{\hspace{4.5em}} $\cup\ \{\alpha\!:\!\V R.\ovl{D} \mid\ (\alpha\!:\,\leq\!0\,R.D) \in \Label(v)\}$ \\
	\mbox{\hspace{4.5em}} $\cup\ \{\alpha\!:\!\V R.D \mid \alpha\!:\!\V S.D \in \Label(v) \textrm{ and } R \sqsubseteq_\mR S\}$ \\
	\mbox{\hspace{4.5em}} $\cup\ \{R^-(b,a) \mid R(a,b) \in \Label(v)\}$ \\
	\mbox{\hspace{4.5em}} $\cup\ \{S(a,b) \mid R(a,b) \in \Label(v) \textrm{ and } R \sqsubseteq_\mR S\}$ \\
	\mbox{\hspace{4.5em}} $\cup\ \{b\!:\!D \mid \{a\!:\!\V R.D, R(a,b)\} \subseteq \Label(v)\}$ \\
	\mbox{\hspace{4.5em}} $\cup\ \{b\!:\!\V R.D \mid \{a\!:\!\V R.D, R(a,b)\} \subseteq \Label(v) \textrm{ and } \Trans{R}\}$ \\
	\mbox{\hspace{4.5em}} $-\ (X \cup \RFormulas(v))$
     \item if $label - \Label(v) \neq \emptyset$ then
	\begin{enumerate}
	\item $\ConToSucc(v,\NonState,\SType(v),label,\RFormulas(v) \cup X,\Null)$
	\item $\Status(v) := \FExpanded$.
	\end{enumerate}
     \end{enumerate}

  \begin{explanation}
  This rule makes a necessary expansion for a non-state $v$ by connecting it to only one successor $w$ which is a copy of $w$ with intuitive changes like:
  \begin{itemize}
  \item if $\alpha\!:\!(D \mand D') \in \Label(v)$ then $\alpha\!:\!(D \mand D')$ in $\Label(w)$ is replaced by $\alpha\!:\!D$ and $\alpha\!:\!D'$ and we remember this by adding it to $\RFormulas(w)$;
  \item if $\{a\!:\!\V R.D, R(a,b)\} \subseteq \Label(v)$ then we add $b\!:\!D$ to $\Label(w)$; and so on.
  \end{itemize}
  Note that $\Label(w) - (\Label(v) \cup \RFormulas(v)) \neq \emptyset$. That is, $w$ contains some ``new'' formulas. This is to guarantee that the local graph of any non-state is acyclic.
\koniec 
  \end{explanation}
  \end{description}

\sssection{The Rules for Keeping Converse Compatibility:}
Rules of this kind are listed below in the decreasing order w.r.t.\ priority: 
  \begin{description}
  \item[\KCCd] If $\tuple{u,v} \in E$ and $\Status(v) = \Incomplete$ then
     \begin{enumerate}
     \item (we must have that $\Type(v) = \State$ and $\Type(u) = \NonState$)
     \item delete the edge $\tuple{u,v}$ from $E$ and re-expand $u$ as follows
     \item if $\FmlsRC(v) \neq \emptyset$ then 
	\begin{enumerate}
	\item $newLabel := \Label(u) \cup \FmlsRC(v)$
	\item $\ConToSucc(u,\NonState,\SType(u),newLabel,\RFormulas(u),\Null)$
	\end{enumerate}
     \item else
	\begin{enumerate}
	\item $newLabel_1 := \Label(u) \cup \{\FmlFB(v)\}$
	\item $\ConToSucc(u,\NonState,\SType(u),newLabel_1,\RFormulas(u),\Null)$
	\item $newLabel_2 := \Label(u) \cup \{\ovl{\FmlFB(v)}\}$
	\item $\ConToSucc(u,\NonState,\SType(u),newLabel_2,\RFormulas(u),\Null)$.
	\end{enumerate}
     \end{enumerate}

  \begin{explanation}
  We have that $v$ is a state with status $\Incomplete$ and it is the only successor of $u$. The first expansion of $u$ (by connecting to $v$) was not a good move and we re-expand $u$ (only once) as follows. We first delete the edge $\tuple{u,v}$. Next, if $\FmlsRC(v) \neq \emptyset$ (i.e., there are formulas that should be added to $v$), then we connect $u$ to a node with label equal to $\Label(u) \cup \FmlsRC(v)$ (this node is a replacement for $v$). If $\FmlsRC(v) = \emptyset$ then the reason of $\Status(v) = \Incomplete$ is that we wanted to have either $\FmlFB(v)$ or its negation in $\Label(v)$. So, in that case we connect $u$ to two successors, one with label $\Label(u) \cup \{\FmlFB(v)\}$ and the other with label $\Label(u) \cup \{\ovl{\FmlFB(v)}\}$. The node $v$ is useful for later: whenever we expand a node $u'$ using the forming-state rule by connecting it to $v$ we know immediately that we should re-expand it. The other nodes in the local graph of $v$ can be deleted to save memory. However, we keep them for our presentation to make the proofs easier. 
\koniec
  \end{explanation}

  \item[\KCCa] If $\Type(v) = \State$, $\Status(v) = \PExpanded$ and only the rules for updating status and the unary static rule were applied to the nodes in the local graph of $v$ that are different from $v$ then
     \begin{enumerate}
     \item $X := \emptyset$
     \item for each node $w$ in the local graph of $v$ do 
	\begin{enumerate}
	\item let $w_0 = \AfterTransPred(w)$ and $\alpha = \CELabelI(w_0)$
	\item $X := X \cup \{\alpha\!:\!C \mid$ there exist $R \in \CELabelR(w_0)$, $\V R^-.C \in \Label(w)$ and $\alpha\!:\!C \notin \FullLabel(v)\}$ 
	\item $X := X \cup \{\alpha\!:\!\V R^-.C \mid$ $R \in \CELabelR(w_0)$, $\Trans{R}$, $\V R^-.C \in \Label(w)$ and $\alpha\!:\!\V R^-.C \notin \FullLabel(v)\}$ 
	\end{enumerate}
     \item if $X \neq \emptyset$ then set $\FmlsRC(v) := X$ and $\Status(v) := \Incomplete$. 
     \end{enumerate}

  \begin{explanation}\label{exp: HGREW}
  Assume that the considered knowledge base is satisfiable and has a model $\mI$ that satisfies $\Label(v)$. Consider the following cases:
     \begin{itemize}
     \item Case $v$ is a simple state and $\alpha = \Null$: Thus, $v$ corresponds to an unnamed individual $y_v \in \Delta^\mI$, while $w$ corresponds to an $R$-successor $y_w$ of $y_v$ in $\mI$. If $\V R^-.C \in \Label(w)$ then: $y_w \in (\V R^-.C)^\mI$ and hence $y_v \in C^\mI$, and for that reason we want to have $C$ in $\FullLabel(v)$ as a requirement to be realized; so, if $C \notin \FullLabel(v)$ then the status of $v$ becomes $\Incomplete$ and we add $C$ to $\FmlsRC(v)$ as a concept required for $v$ for converse compatibility. If $\V R^-.C \in \Label(w)$ and $R$ is a transitive role then: $R^-$ is also transitive and we also have $y_v \in (\V R^-.C)^\mI$; so, analogously, if $\V R^-.C \notin \FullLabel(v)$ then the status of $v$ becomes $\Incomplete$ and we add $\V R^-.C$ to $\FmlsRC(v)$.  

     \item Case $v$ is a complex state and $\alpha = a$: Thus, $v$ corresponds to an ABox with that individual~$a$. The node $w$ corresponds to an $R$-successor $y_w$ of $a^\mI$ in $\mI$. If $\V R^-.C \in \Label(w)$ then: $y_w \in (\V R^-.C)^\mI$ and hence $a^\mI \in C^\mI$, and for that reason we want to have $a\!:\!C$ in $\FullLabel(v)$ as a requirement to be realized; so, if $a\!:\!C \notin \FullLabel(v)$ then the status of $v$ becomes $\Incomplete$ and we add $a\!:\!C$ to $\FmlsRC(v)$ as an assertion required for $v$ for converse compatibility. If $a\!:\!(\V R^-.C) \in \Label(w)$ and $R$ is a transitive role then: $R^-$ is also transitive and we also have $a^\mI \in (\V R^-.C)^\mI$; so, analogously, if $a\!:\!(\V R^-.C) \notin \FullLabel(v)$ then the status of $v$ becomes $\Incomplete$ and we add $a\!:\!(\V R^-.C)$ to $\FmlsRC(v)$.
\koniec
     \end{itemize}
  \end{explanation}

  \item[\KCCb] If $\Type(u) = \State$, $v_0$ is a successor of $u$ with $\CELabelI(v_0) = \alpha$ and $R \in \CELabelR(v_0)$, and $v$ is a node in the local graph of $v_0$ with $\V R^-.C \in \Label(v)$ and $\Status(v) \neq \Unsat$ then
     \begin{enumerate}
     \item if $\alpha\!:\!\ovl{C} \in \FullLabel(u)$ then $\Status(v) := \Unsat$
     \item else if $\Trans{R}$ and $\alpha\!:\!\E R^-.\ovl{C} \in \FullLabel(u)$ then $\Status(v) := \Unsat$
     \item else if $(\alpha\!:\!C) \notin \FullLabel(u)$ and $\Status(u) \in \{\PExpanded$, $\FExpanded\}$ then $\Status(u) := \Incomplete$ and $\FmlFB(u) := (\alpha\!:\!C)$
     \item else if $\Trans{R}$ and $(\alpha\!:\!\V R^-.C) \notin \FullLabel(u)$ and $\Status(u) \in \{\PExpanded$, $\FExpanded\}$ then set $\Status(u) := \Incomplete$ and $\FmlFB(u) := (\alpha\!:\!\V R^-.C)$.
     \end{enumerate}

  \begin{explanation}
  Assume that the considered knowledge base is satisfiable and has a model $\mI$ that satisfies $\Label(u)$. We consider here only the case when $u$ is a simple state and $\alpha = \Null$. (The case when $u$ is a complex state and $\alpha$ is a named individual is similar, cf.\ Explanation~\ref{exp: HGREW}.) Thus, $u$ corresponds to an unnamed individual $y_u \in \Delta^\mI$. In the constructed tableau, there is an expansion path from $u$ via $v_0$ to $v$. The question is: whether that expansion path should be taken into account, and if we should consider that expansion path, what should be done for converse compatibility. Note that: if $\CELabelT(v_0) = \CQF$ and $(x_{v_0} = 0) \in \ILConstraints(u)$ then $v_0$ is not used for the construction of $\mI$; and if the path from $v_0$ to $v$ uses some non-unary static expansions then $v$ is just one of possible ``expansions'' of $v_0$. 
      \begin{itemize}
      \item If $v$ corresponds to an $R$-successor $y_v$ of $y_u$ in $\mI$ then $y_v \in (\V R^-.C)^\mI$ (since $\V R^-.C \in \Label(v)$) and hence $y_u \in C^\mI$, which means we should require $C \in \FullLabel(u)$. So, if $\ovl{C} \in \FullLabel(u)$ then the expansion path from $u$ via $v_0$ to $v$ cannot be used and we set $\Status(v) := \Unsat$ to mark that $v$ cannot be used to realize the requirements of $u$. 

      \item Else if $R$ is a transitive role and $v$ corresponds to an $R$-successor $y_v$ of $y_u$ in $\mI$ then we also have $y_u \in (\V R^-.C)^\mI$, which means we should require $\V R^-.C \in \FullLabel(u)$. So, if $R$ is transitive and $\E R^-.\ovl{C} \in \FullLabel(u)$ (note that $\E R^-.\ovl{C}$ is the negation of $\V R^-.C$) then the expansion path from $u$ via $v_0$ to $v$ cannot be used and we set $\Status(v) := \Unsat$ to mark that $v$ cannot be used to realize the requirements of $u$. 

      \item Else if $\V R^-.C \in \Label(v)$ and $\{C,\ovl{C}\} \cap \FullLabel(u) = \emptyset$ then we would like to have either $C$ or $\ovl{C}$ in $\Label(u)$ and therefore set $\Status(u) := \Incomplete$ and $\FmlFB(u) := C$ (as a concept for branching on at $u$).  

      \item Else if $R$ is transitive, $\V R^-.C \in \Label(v)$ and $\{\V R^-.C,\E R^-.\ovl{C}\} \cap \FullLabel(u) = \emptyset$ then we would like to have either $\V R^-.C$ or $\E R^-.\ovl{C}$ in $\Label(u)$ and therefore set $\Status(u) := \Incomplete$ and $\FmlFB(u) := \V R^-.C$ (as a concept for branching on at $u$).  
\koniec
      \end{itemize}
  \end{explanation}

  \item[\KCCc] If $u$ is a state with $\Status(u) \in \{\PExpanded$, $\FExpanded\}$, $v_0$ is a successor of $u$ with $\CELabelI(v_0) = \alpha$ and $R \in \CELabelR(v_0)$, $v$ is a node in the local graph of $v_0$ such that $\Status(v) \neq \Unsat$ and
     \begin{itemize}
     \item either $\Label(v)$ contains $\leq\!n\,R^-.C$,  
     \item or $\Label(v)$ contains $\geq\!n\,R^-.C$ or $\E R^-.C$, where $R$ is a numeric role, 
     \end{itemize} 
  and $\{(\alpha\!:\!C),(\alpha\!:\!\ovl{C})\} \cap \FullLabel(u) = \emptyset$\\ then set $\Status(u) := \Incomplete$ and $\FmlFB(u) := (\alpha\!:\!C)$.

  \begin{explanation}\label{exp: JHRES}
  This rule deals with the case when there is a lack of information at $u$ for deciding how to satisfy the number restrictions of $v$. We want to have either $\alpha\!:\!C$ or $\alpha\!:\!\ovl{C}$ in $\FullLabel(u)$. So, we set $\Status(u) := \Incomplete$ and $\FmlFB(u) := (\alpha\!:\!C)$ (as a formula for branching on at $u$).
\koniec
  \end{explanation}

  \item[\KCCcp] If $u$ is a state with $\Status(u) = \FExpanded$, $v_0$ is a successor of $u$ with $\CELabelT(v_0) = \CQF$, $R \in \CELabelR(v_0)$ and $\CELabelI(v_0) = \alpha$, and $v$ is a node in the local graph of $v_0$ such that
     \begin{itemize}
     \item $\Status(v) \neq \Unsat$, $\Label(v)$ contains $\leq\!m\,R^-.C$ as well as $\geq\!n\,S^-.D$ or $\E S^-.D$, 
     \item $S \sqsubseteq_\mR R$, $S \notin \CELabelR(v_0)$, $\alpha\!:\!C \in \FullLabel(u)$, $\alpha\!:\!\ovl{D} \notin \FullLabel(u)$, 
     \item $\{(\alpha\!:\!\E S.\top_R),(\alpha\!:\!\V S.\bot)\} \cap \Label(u) = \emptyset$
     \end{itemize}
then set $\Status(u) := \Incomplete$ and $\FmlFB(u) := (\alpha\!:\!\E S.\top_R)$.\\
(Here, $\top_R$ is $\top$ annotated by $R$. Semantically, $\top_R$ is equivalent to $\top$, i.e., $\top_R^\mI = \top^\mI = \Delta^\mI$ for every interpretation~$\mI$. Besides, let $\ovl{\top_R} = \bot$ and hence $\ovl{\E S.\top_R} = \V S.\bot$.)

  \begin{explanation}\label{exp: IUSDS}
  Assume that $\Label(v)$ contains $\varphi$ which is either $\geq\!n\,S^-.D$ (with $n \geq 1$) or $\E S^-.D$. To realize the requirement $\varphi$ at $v$ we want to know whether $u$ can be used for that purpose. That is, we want ($S \in \CELabelR(v_0)$ and $\alpha\!:\!D \in \FullLabel(u)$) or $S \notin \CELabelR(v_0)$ or $\alpha\!:\!\ovl{D} \in \FullLabel(u)$. If $S \in \CELabelR(v_0)$ then the will to have either $\alpha\!:\!D$ or $\alpha\!:\!\ovl{D}$ in $\FullLabel(u)$ can be realized by the rule $\KCCc$. If $\alpha\!:\!\ovl{D} \in$ $\FullLabel(u)$ then $u$ cannot be used to realize the requirement $\varphi$ at $v$. Consider the case when $S \notin \CELabelR(v_0)$ and $\alpha\!:\!\ovl{D} \notin \FullLabel(u)$. If $v$ was treated as a state then to realize the requirement $\varphi$ at $v$ we would create an $S^-$-successor $w$ of $v$ and put $D$ to $\Label(w)$. As $S \sqsubseteq_\mR R$, the node $w$ is also an $R^-$-successor of $v$. Since $\leq\!m\,R^-.C \in \Label(v)$ and $\alpha\!:\!C \in \FullLabel(u)$, there may be the need to merge $w$ to $u$. Such merging would add $S$ to $\CELabelR(v_0)$. However, our method does not merge nodes explicitly. Roughly speaking, we want to decide whether to add $S$ to $\CELabelR(v_0)$ or not. To solve the problem, we set $\Status(u) := \Incomplete$ and $\FmlFB(u) := (\alpha\!:\!\E S.\top_R)$. Later, if a state $u'$ is a ``completion'' of $u$, then either $\alpha\!:\!\E S.\top_R \in \Label(u')$ or $\alpha\!:\!\V S.\bot \in \Label(u')$. If $\alpha\!:\!\V S.\bot \in \Label(u')$ then $u'$ will not have any $S$-successor. If $\alpha\!:\!\E S.\top_R \in \Label(u')$ then $u'$ will have some $S$-successor $v'$ with $\top_R \in \Label(v')$, and the occurrence of $\top_R$ in $\Label(v')$ guarantees that the possibility of merging $v'$ to any $R$-successor of $u'$ will be considered by the transitional full-expansion rule. 
\koniec
  \end{explanation}
  \end{description}

\sssection{The Non-unary Static Expansion Rule:}
  \begin{description}
  \item[\NUS] If $\Type(v) = \NonState$ and $\Status(v) = \Unexpanded$ then
     \begin{enumerate}
     \item\label{item: JHREA} if $\alpha\!:\!(C \mor D) \in \Label(v)$ and $\{\alpha\!:\!C, \alpha\!:\!D\} \cap \FullLabel(v) = \emptyset$ then
	\begin{enumerate}
	\item let $X = \Label(v) - \{\alpha\!:\!(C \mor D)\}$ 
	\item let $Y = \RFormulas(v)\cup\{\alpha\!:\!(C \mor D)\}$
	\item $\ConToSucc(v,\NonState,\SType(v),X\cup\{\alpha\!:\!C\},Y,\Null)$
	\item $\ConToSucc(v,\NonState,\SType(v),X\cup\{\alpha\!:\!D\},Y,\Null)$
	\item $\Status(v) := \FExpanded$
	\end{enumerate}

     \begin{explanation}
     This subrule deals with syntactic branching on $\alpha\!:\!(C \mor D) \in \Label(v)$. We expand $v$ by connecting it to two successors $w_1$ and $w_2$, whose labels are the label of $v$ with $\alpha\!:\!(C \mor D)$ replaced by $\alpha\!:\!C$ or $\alpha\!:\!D$, respectively. The formula $\alpha\!:\!(C \mor D)$ is put into both $\RFormulas(w_1)$ and $\RFormulas(w_2)$. The expansion is done only when both $w_1$ and $w_2$ have a larger $\FullLabel$ than $v$. 
\koniec
     \end{explanation}

     \item\label{item: HGW3A} else if $\SType(v) = \Complex$, $R(a,b) \in \Label(v)$ and 
       \begin{itemize}
       \item either $\Label(v)$ contains $a\!:\!(\leq\!n\,R.C)$, 
       \item or $\Label(v)$ contains $a\!:\!(\geq\!n\,R.C)$ or $a\!:\!(\E R.C)$, where $R$ is a numeric role, 
       \end{itemize}
     and $\{b\!:\!C, b\!:\!\ovl{C}\} \cap \FullLabel(v) = \emptyset$ then 
	\begin{enumerate}
	\item $\ConToSucc(v, \NonState, \Complex, \Label(v)\cup\{b\!:\!C\}, \RFormulas(v), \Null)$
	\item $\ConToSucc(v, \NonState, \Complex, \Label(v)\cup\{b\!:\!\ovl{C}\}, \RFormulas(v), \Null)$
	\item $\Status(v) := \FExpanded$
	\end{enumerate}

     \begin{explanation}
  This subrule deals with the case when there is a lack of information about $b$ for deciding how to satisfy the number restrictions about $a$ (cf.\ Explanation~\ref{exp: JHRES}). We want to have either $b\!:\!C$ or $b\!:\!\ovl{C}$ in $\FullLabel(v)$. So, we expand $v$ by semantic branching: we connect it to two successors, one with label $\Label(v)\cup\{b\!:\!C\}$ and the other with label $\Label(v)\cup\{b\!:\!\ovl{C}\}$. The expansion is done only when both the successors have a larger $\FullLabel$ than~$v$.
\koniec
     \end{explanation}

     \item\label{item: JHEAA} else if $\SType(v) = \Complex$, $\{a\!:\!(\leq\!n\,R.C)$, $R(a,b)$, $R(a,b')$, $b\!:\!C$, $b'\!:\!C\} \subseteq$ $\FullLabel(v)$, $b \neq b'$ and $\{b \not\doteq b', b' \not\doteq b\} \cap \Label(v) = \emptyset$ then\footnote{Fix a linear order between named individuals. Then we can also assume that $b$ is less than $b'$ in that order.}
	\begin{enumerate}
	\item let $X$ be the set obtained from $\Label(v)$ by replacing every occurrence of $b'$ not in $\doteq$~expressions by $b$
	\item let $Y$ be the set obtained from $\RFormulas(v)$ by replacing every occurrence of $b'$ by $b$
	\item\label{item: HJSDA} $\ConToSucc(v,\NonState,\Complex,X \cup \{b \doteq b', b' \doteq b\},Y,\Null)$
	\item $\ConToSucc(v$, $\NonState$, $\Complex$, $\Label(v)\cup\{b \not\doteq b',b' \not\doteq b\}$, $\RFormulas(v)$, $\Null)$
	\item $\Status(v) := \FExpanded$
	\end{enumerate}

     \begin{explanation}
  This subrule deals with the case when there is a lack of information about whether $b$ and $b'$ denote the same individual for deciding how to satisfy the number restrictions about $a$. We expand $v$ by semantic branching: either $b$ and $b'$ denote the same individual or they do not. Technically, we connect $v$ to two successors with appropriate contents. 
\koniec
     \end{explanation}

     \item\label{item: OSJRS} else if $\SType(v) = \Complex$, $\{a\!:\!(\leq\!m\,R.C)$, $R(a,b)\} \subseteq \Label(v)$, $\Label(v)$ contains \mbox{$a\!:\!(\geq\!n\,S.D)$} or $a\!:\!\E S.D$ with $S \sqsubseteq_\mR R$, and $\{S(a,b), \lnot S(a,b)\} \cap \Label(v) = \emptyset$ then
	\begin{enumerate}
	\item $\ConToSucc(v,\NonState,\Complex,\Label(v) \cup \{S(a,b)\},\RFormulas(v),\Null)$
	\item $\ConToSucc(v,\NonState,\Complex,\Label(v) \cup \{\lnot S(a,b)\},\RFormulas(v),\Null)$
	\item $\Status(v) := \FExpanded$.
	\end{enumerate}

     \begin{explanation}
  This subrule deals with the case when there is a lack of information for deciding how to satisfy the number restrictions about $a$ (cf.\ the rule~\KCCcp and Explanation~\ref{exp: IUSDS}). We want to decide whether $b$ is an $S$-successor of $a$ or not. So, we expand $v$ by semantic branching: we connect it to two successors, one with label containing $S(a,b)$ and the other with label containing $\lnot S(a,b)$. The expansion is done only when both the successors have a larger label than~$v$.
\koniec
     \end{explanation}
     \end{enumerate}
  \end{description}

\sssection{The Forming-State Rules:}
  \begin{description}
  \item[\FSa] If $\Type(v) = \NonState$, $\SType(v) = \Simple$ and $\Status(v) = \Unexpanded$ then
     \begin{enumerate}
     \item let $u = \StatePred(v)$, $v_0 = \AfterTransPred(v)$ and $\alpha = \CELabelI(v_0)$
     \item set $X := \Label(v)$
     \item for each $(\leq\!n\,R.D) \in \Label(v)$ do
	\begin{enumerate}
	\item if $R^- \in \CELabelR(v_0)$ and $(\alpha\!:\!D) \in \FullLabel(u)$ then add $\preceq\!(n-1)\,R.D$ to $X$ 
	\item else add $\preceq\!n\,R.D$ to $X$
	\end{enumerate}
     \item for each $(\geq\!n\,R.D) \in \Label(v)$ with $n \geq 2$ do
	\begin{enumerate}
	\item if $R^- \in \CELabelR(v_0)$ and $(\alpha\!:\!D) \in \FullLabel(u)$ then add $\succeq\!(n-1)\,R.D$ to $X$ 
	\item else add $\succeq\!n\,R.D$ to $X$
	\end{enumerate}
     \item for each $(\geq\!1\,R.D)$ or $\E R.D$ from $\Label(v)$, where $R$ is a numeric role, do
	\begin{enumerate}
	\item if $R^- \notin \CELabelR(v_0)$ or $(\alpha\!:\!D) \notin \FullLabel(u)$ then add $\succeq\!1\,R.D$ to $X$ 
	\end{enumerate}
     \item $\ConToSucc(v,\State,\Simple,X,\RFormulas(v),\Null)$
     \item $\Status(v) := \FExpanded$.
     \end{enumerate}

  \begin{explanation}
  When the rules \UPS, \US and \KCC are not applicable (to any node) and the rule \NUS is not applicable to $v$, we apply this forming-state rule to $v$ by connecting it to a simple state $w$. Assume that the considered knowledge base is satisfiable and has a model $\mI$ such that the nodes $v$ and $w$ are used for the construction of $\mI$. In the constructed graph, such a state $w$ is globally cached and may have many predecessors. However, the elements $y_w$ of the domain of $\mI$ that correspond to $w$ are nodes in disjoint trees (cf.\ Remark~\ref{remark: main intuition}), and each of them has only one predecessor $y_u$, which corresponds to $u$. When computing contents for $w$ we put into $\Label(w)$ the requirements from $\Label(v)$ after an appropriate modification that takes into account the relationship between $y_w$ and $y_u$, i.e., the relationship between $v$ and $u$ via $v_0$. For example, if $(\leq\!n\,R.D) \in \Label(v)$, $R^- \in \CELabelR(v_0)$ and $(\alpha\!:\!D) \in \FullLabel(u)$ then $y_w \in \Delta^\mI$ (which corresponds to all $w$, $v$, $v_0$) already has the $R$-successor $y_u \in D^\mI$, and we have to guarantee only that $y_w$ has $n-1$ other $R$-successors satisfying $D$ (i.e., belonging to $D^\mI$), and that is why we add to $\Label(w)$ the requirement $\preceq\!(n-1)\,R.D$. Notice the use of $\preceq$ instead of $\leq$. Also note that we can assume $\Status(u) \neq \Incomplete$ (otherwise, there is no sense for expanding~$v$) and, as the rule for keeping converse compatibility is not applicable, in that case either $\alpha\!:\!D$ or $\alpha\!:\!\ovl{D}$ belongs to $\FullLabel(u)$. 
\koniec
  \end{explanation}

  \item[\FSb] If $\Type(v) = \NonState$, $\SType(v) = \Complex$ and $\Status(v) = \Unexpanded$ then
     \begin{enumerate}
     \item set $X := \Label(v)$
     \item for each $a\!:\!(\leq\!n\,R.D) \in \Label(v)$ do
	\begin{enumerate}
	\item let $m = \sharp\{b \mid \{R(a,b), b\!:\!D\} \subseteq \FullLabel(v)\}$
	\item add $a\!:\!(\preceq\!(n-m)\,R.D)$ to $X$
	\end{enumerate}
     \item for each $(a\!:\!C) \in \Label(v)$, where $C$ is $\geq\!n\,R.D$ or $\E R.D$ and $R$ is a numeric role, do
	\begin{enumerate}
	\item if $C = \E R.D$ then let $n = 1$
	\item let $m = \sharp\{b \mid \{R(a,b), b\!:\!D\} \subseteq \FullLabel(v)\}$
	\item if $n > m$ then add $a\!:\!(\succeq\!(n-m)\,R.D)$ to $X$
	\end{enumerate}
     \item $\ConToSucc(v,\State,\Complex,X,\RFormulas(v),\Null)$
     \item $\Status(v) := \FExpanded$.
     \end{enumerate}

  \begin{explanation}
When the rules \UPS, \US, \KCC and \NUS are not applicable to the complex non-state $v$, we apply this forming-state rule to $v$ by connecting it to a complex state $w$. When computing contents for $w$ we put into $\Label(w)$ the requirements from $\Label(v)$ after an appropriate modification that takes into account the assertions in $\Label(v)$ that represent the relationship between named individuals. For example, if $a\!:\!(\leq\!n\,R.D) \in \Label(v)$ and there are $m$ pairwise different individuals $b_1,\ldots,b_m$ such that $\{R(a,b_i), b_i\!:\!D \mid 1 \leq i \leq m\} \subseteq \FullLabel(v)$ then we add to $\Label(w)$ the requirements $a\!:\!(\preceq\!(n-m)\,R.D)$. Notice the use of $\preceq$ instead of $\leq$. Note that, since the rule \NUS is not applicable to $v$, we must have that $(b_i \not\doteq b_j) \in \Label(v)$ for any pair $i \neq j$, and for any individual $b$ such that $R(a,b) \in \Label(v)$, either $b\!:\!D$ or $b\!:\!\ovl{D}$ must belong to $\FullLabel(v)$. When expanding $w$ we will not have to pay attention to the relationship between the individuals occurring in $\Label(w)$. 
\koniec
  \end{explanation}
  \end{description}

\sssection{The Transitional Partial-Expansion Rule:}
  \begin{description}
  \item[\TP] If $\Type(v) = \State$ and $\Status(v) = \Unexpanded$ then
     \begin{enumerate}
     \item for each $(\alpha\!:\!\E R.D) \in \Label(v)$, where $R$ is a non-numeric role, do
       \begin{enumerate}
	\item $X := \{D\} \cup \{D' \mid \alpha\!:\!\V R.D' \in \Label(v)\}\ \cup$\\
		\mbox{\hspace{2.55em}}$\{\V S.D' \mid \alpha\!:\!\V S.D' \in \Label(v), R \sqsubseteq_\mR S$ and $\Trans{S}\} \cup \mT$
	\item $ceLabel := \langle \TUS, \{S \mid R \sqsubseteq_\mR S\}, \alpha\rangle$
	\item $\NewSucc(v,\NonState,\Simple,X,\emptyset,ceLabel)$
       \end{enumerate}

     \item $\Status(v) := \PExpanded$.
     \end{enumerate}

  \begin{explanation}
  To realize a requirement $\alpha\!:\!\E R.D$ at $v$, where $R$ is a non-numeric role, we connect $v$ to a new simple non-state $w$ with appropriate contents as shown in the rule.
\koniec
  \end{explanation}
  \end{description}

\sssection{The Transitional Full-Expansion Rule:}
  \begin{description}
  \item[\TF] If $\Type(v) = \State$ and $\Status(v) = \PExpanded$ then
     \begin{enumerate}
     \item $\mE := \emptyset$, $\mE' := \emptyset$
     \item\label{item: HJFDU} for each $(\alpha\!:\,\succeq\!n\,R.D) \in \Label(v)$ do
       \begin{enumerate}
	\item $X := \{S \mid R \sqsubseteq_\mR S\}$
	\item $Y := \{D\} \cup \{D' \mid \alpha\!:\!\V R.D' \in \Label(v)\}\ \cup$\\
		\mbox{\hspace{2.55em}}$\{\V S.D' \mid \alpha\!:\!\V S.D' \in \Label(v), R \sqsubseteq_\mR S$ and $\Trans{S}\} \cup \mT$
	\item $\mE := \mE \cup \{\tuple{X,Y,\alpha}\}$
       \end{enumerate}
     \item\label{item: HGDFW} for each $\alpha\!:\!(\preceq\!n\,R.C) \in \Label(v)$ do
	\begin{enumerate}
	\item for each $\tuple{X,Y,\alpha} \in \mE$ do
	   \begin{enumerate}
	   \item if $R \in X$ and $\{C,\ovl{C}\} \cap Y = \emptyset$ then 
		$\mE' := \mE' \cup \{\tuple{X,Y \cup \{C\},\alpha},\tuple{X,Y\cup\{\ovl{C}\},\alpha}\}$\\
	(i.e., $\tuple{X,Y,\alpha}$ is replaced by $\tuple{X,Y \cup \{C\},\alpha}$ and $\tuple{X,Y\cup\{\ovl{C}\},\alpha}$)
	   \item else $\mE' := \mE' \cup \{\tuple{X,Y,\alpha}\}$
	   \end{enumerate}
	\item $\mE := \mE'$, $\mE' := \emptyset$
	\end{enumerate}

     \item\label{item: HMFHS} repeat
	\begin{enumerate}
	\item for each $\alpha\!:\!(\preceq\!n\,R.C) \in \Label(v)$, $\tuple{X,Y,\alpha} \in \mE$ and $\tuple{X',Y',\alpha} \in \mE$ such that $R \in X$, $C \in Y$, $R \in X'$, $C \in Y'$, $\tuple{X \cup X', Y \cup Y', \alpha} \notin \mE$ and $Y \cup Y'$ does not contain any pair of the form $\varphi$, $\ovl{\varphi}$ do add $\tuple{X \cup X', Y \cup Y', \alpha}$ to $\mE$ (i.e., the merger of $\tuple{X,Y,\alpha}$ and $\tuple{X',Y',\alpha}$ is added to $\mE$)
	\item\label{item: YDMAT} for each $\tuple{X,Y,\alpha} \in \mE$, $\tuple{X',Y',\alpha} \in \mE$ and $R \in X \cap X'$ such that $\top_R \in Y \cup Y'$, $\tuple{X \cup X', Y \cup Y', \alpha} \notin \mE$ and $Y \cup Y'$ does not contain any pair of the form $\varphi$, $\ovl{\varphi}$ do add $\tuple{X \cup X', Y \cup Y', \alpha}$ to $\mE$
	\end{enumerate}
	until no tuples were added to $\mE$ during the last iteration

     \item\label{item: JHFSA} for each $\tuple{X,Y,\alpha} \in \mE$ do
	\begin{enumerate}
	\item $\NewSucc(v,\NonState,\Simple,Y,\emptyset,\tuple{\CQF,X,\alpha})$
	\end{enumerate}

     \item\label{item: JHDSA} let $W = \{w \mid \tuple{v,w} \in E \textrm{ and } \CELabelT(w) = \CQF\}$
     \item $\ILConstraints(v) := \{ x_w \geq 0 \mid w \in W \}$
     \item\label{item: JEROS} for each $(\alpha\!:\!C) \in \Label(v)$ do
       \begin{enumerate}
       \item if $C$ is of the form $\succeq\!n\,R.D$ 
	then add to $\ILConstraints(v)$ the constraint\\ $\sum \{ x_w \mid w \in W, \CELabelI(w) = \alpha, R \in \CELabelR(w), D \in \Label(w)\} \geq n$

       \item if $C$ is of the form $\preceq\!n\,R.D$ 
	then add to $\ILConstraints(v)$ the constraint\\ $\sum \{ x_w \mid w \in W, \CELabelI(w) = \alpha, R \in \CELabelR(w), D \in \Label(w)\} \leq n$
	\end{enumerate}

     \item $\Status(v) := \FExpanded$.
     \end{enumerate}

  \begin{explanation}
  To satisfy a requirement $\varphi = (\alpha\!:\,\succeq\!n\,R.C) \in \Label(v)$, one can first create a successor $w_\varphi$ of $v$ specified by the tuple $\tuple{X,Y,\alpha}$ computed at the step~\ref{item: HJFDU}, where $X$ presents $\CELabelR(w_\varphi)$, $Y$ presents $\Label(w_\varphi)$ and $\alpha$ presents $\CELabelI(w_\varphi)$, and then clone $w_\varphi$ to create $n$ successors for $v$ (or only record the intention somehow). The label of $w_\varphi$ contains only formulas necessary for realizing the requirement $\alpha\!:\!\E R.C$ and related ones of the form $\alpha\!:\!\V R'.C'$ at $v$. 
  To satisfy requirements of the form $\alpha\!:\preceq\!n'\,R'.C'$ at $v$, where $R \sqsubseteq_\mR R'$, we tend to use only copies of $w_\varphi$ extended with either $C'$ or $\ovl{C'}$ (for easy counting) as well as the mergers of such extended nodes. 
  So, we first start with the set $\mE$ constructed at the step~\ref{item: HJFDU}, which consists of tuples with information about successors to be created for $v$. 
  We then modify $\mE$ by taking necessary extensions of the nodes (see the step~\ref{item: HGDFW}). After that we continue modifying $\mE$ by adding to it also appropriate mergers of nodes (see the step~\ref{item: HMFHS}). The merging specified at the step~\ref{item: YDMAT} corresponds to the rule~\KCCcp with Explanation~\ref{exp: IUSDS}. Successors for $v$ are created at the step~\ref{item: JHFSA}. The number of copies of a node $w$ that are intended to be used as successors of $v$ is represented by a variable $x_w$ (we will not actually create such copies). The set $\ILConstraints(v)$ consisting of appropriate constraints about such variables are set at the steps \ref{item: JHDSA}-\ref{item: JEROS}. 
\koniec
  \end{explanation}
  \end{description}


\subsection{Properties of \CSHIQ-Tableaux}

Define the size of a knowledge base $\KB = \tuple{\mR,\mT,\mA}$ to be the number of bits used for the usual sequential representation of $\KB$. It is greater than the number of symbols occurring in $\KB$. If $\Size$ is the size of $\KB$ and $\leq\!n\,R.C$ or $\geq\!n\,R.C$ is a number restriction occurring in $\KB$ then:
\begin{itemize}
\item when numbers are coded in unary we have that $n \leq \Size$,  
\item when numbers are coded in binary we have that $n \leq 2^\Size$.  
\end{itemize}

\newcommand{\LemmaComplexity}{Let $\tuple{\mR,\mT,\mA}$ be a knowledge base in NNF of the logic \SHIQ and let $\Size$ be the size of $\tuple{\mR,\mT,\mA}$. Then a \CSHIQ-tableau for $\tuple{\mR,\mT,\mA}$ can be constructed in (at most) exponential time in~$\Size$ in the following cases:
\begin{enumerate}
\item numbers are coded in unary, 
\item numbers are coded in binary and, for any concept $\leq\!n\,R.C$ occurring in $\tuple{\mR,\mT,\mA}$, $n \leq \Size$,  
\item numbers are coded in binary and, for any concept $\geq\!n\,R.C$ occurring in $\tuple{\mR,\mT,\mA}$, $n \leq \Size$. 
\end{enumerate}
} 
\begin{lemma}[Complexity]\label{lemma: Complexity}
\LemmaComplexity
\end{lemma}

\begin{theorem}[Soundness and Completeness]
\label{theorem: s-c}
Let $\tuple{\mR,\mT,\mA}$ be a knowledge base in NNF of the logic \SHIQ and $G = \tuple{V,E,\nu}$ be an arbitrary \CSHIQ-tableau for $\tuple{\mR,\mT,\mA}$. Then $\tuple{\mR,\mT,\mA}$ is satisfiable iff $\Status(\nu) \neq \Unsat$. 
\end{theorem}

See the next section for the proofs of the above lemma and theorem.

To check satisfiability of $\tuple{\mR,\mT,\mA}$ one can construct a \CSHIQ-tableau for it, then return ``no'' when the root of the tableau has status $\Unsat$, or ``yes'' in the other cases. We call this the {\em \CSHIQ-tableau decision procedure}. The corollary given below immediately follows from Theorem~\ref{theorem: s-c} and Lemma~\ref{lemma: Complexity}.

\begin{corollary}
The \CSHIQ-tableau decision procedure has \EXPTIME complexity when numbers are coded in unary.
\end{corollary}


\section{Proofs}
\label{section: proofs}

\subsection{Complexity}

Let $\Size$ be the size of $\tuple{\mR,\mT,\mA}$. 
Define $\closure(\mR,\mT,\mA)$ to be the smallest set $\Gamma$ of formulas such that:
\begin{enumerate}
\item\label{item HGAAD 1} all concepts (and subconcepts) used in $\tuple{\mR,\mT,\mA}$ belong to $\Gamma$,
\item\label{item HGAAD 2} if $R$, $S$ are numeric roles and $S \sqsubseteq_\mR R$ then $\E S.\top_R$, $\V S.\bot$, $\top_R$ and $\bot$ belong to $\Gamma$, 
\item\label{item HGAAD 3} if $\V S.C \in \Gamma$ and $R \sqsubseteq_\mR S$ then $\V R.C \in \Gamma$,
\item\label{item HGAAD 4} if $\leq\!0\,R.C \in \Gamma$ then $\V R.\ovl{C} \in \Gamma$,
\item\label{item HGAAD 5} if $C \in \Gamma$ and $C$ is not of the form $\preceq\!n\,R.C$ nor $\succeq\!n\,R.C$ then $\ovl{C} \in \Gamma$,  
\item\label{item HGAAD 6} if $\E R.C \in \Gamma$ and $R$ is a numeric role then $\succeq\!1\,R.C \in \Gamma$,  
\item\label{item HGAAD 7} if $\geq\!n\,R.C \in \Gamma$, $0 \leq m \leq \Size$ and $m < n$ then $\succeq\!(n-m)\,R.C \in \Gamma$,  
\item\label{item HGAAD 8} if $\leq\!n\,R.C \in \Gamma$, $0 \leq m \leq \Size$ and $m \leq n$ then $\preceq\!(n-m)\,R.C \in \Gamma$,  
\item\label{item HGAAD 9} all assertions of $\mA$ belong to $\Gamma$, 
\item\label{item HGAAD 10} if $C \in \Gamma$ and $a$ is an individual occurring in $\Gamma$ then $a\!:\!C \in \Gamma$, 
\item\label{item HGAAD 11} if $b$ and $b'$ are individuals occurring in $\Gamma$ then $b \doteq b'$ and $b \not\doteq b'$ belong to $\Gamma$, 
\item\label{item HGAAD 12} if $R(a,b) \in \Gamma$ then $R^-(b,a)$ and $\lnot R(a,b)$ belong to $\Gamma$, 
\item\label{item HGAAD 13} if $R(a,b) \in \Gamma$ and $R \sqsubseteq_\mR S$ then $S(a,b) \in \Gamma$,
\item\label{item HGAAD 14} if $R(a,b) \in \Gamma$ and $S \sqsubseteq_\mR R$ then $S(a,b) \in \Gamma$.
\end{enumerate}

\begin{lemma}
The number of formulas of $\closure(\mR,\mT,\mA)$ is of rank $O(\Size^3)$, where $\Size$ is the size of $\tuple{\mR,\mT,\mA}$. 
\end{lemma}

\begin{proof}
The set $\Gamma = \closure(\mR,\mT,\mA)$ can be constructed by initializing $\Gamma$ according to the items~\ref{item HGAAD 1} and \ref{item HGAAD 9}, and then repeatedly applying the rules stated in the remaining items of the list. After initialization the set $\Gamma$ has $O(\Size)$ formulas. The rules in the items~\ref{item HGAAD 2}-\ref{item HGAAD 8} add $O(\Size^2)$ formulas to $\Gamma$.  
The rule in the item~\ref{item HGAAD 10} adds $O(\Size^3)$ formulas to $\Gamma$ (as $\Gamma$ may contains $O(\Size^2)$ concepts and there may be $O(\Size)$ named individuals). The rules in the items~\ref{item HGAAD 11}-\ref{item HGAAD 14} add $O(\Size^2)$ formulas to $\Gamma$. Thus, at the end, $\Gamma$ is of rank $O(\Size^3)$. 
\myEnd
\end{proof}

We recall below Lemma~\ref{lemma: Complexity} before presenting its proof. 

\medskip

\noindent\textbf{Lemma~\ref{lemma: Complexity}.} {\em \LemmaComplexity}
\begin{proof}
Let's construct any \CSHIQ-tableau $G = \tuple{V,E,\nu}$ for $\tuple{\mR,\mT,\mA}$. 

Let $\Size'$ be the number of formulas of $\closure(\mR,\mT,\mA)$. We have $\Size' = O(\Size^3)$. 
For each $v \in V$, $\Label(v) \subseteq \closure(\mR,\mT,\mA)$. Since states of $G$ are cached, it follows that $G$ has no more than $2^{\Size'}$ states. Each state has no more than $O(2^\Size \cdot 2^{\Size'} \cdot \Size)$ successors (since each successor created by the transitional full-expansion rule is characterized by a tuple $\tuple{X,Y,\alpha}$, where $X$ is a set of roles, $Y$ is a set of concepts, and $\alpha$ is $\Null$ or a named individual).\footnote{The bound can be made tighter, e.g., using $O(\Size^2)$ instead of $\Size'$.} If $v$ is a successor of a state, then nodes in the local graph of $v$ are cached and hence there are no more than $2^{\Size'}$ of them. Therefore, $G$ has $O(2^{f(\Size)})$ nodes, where $f(\Size)$ is a polynomial of $\Size$. 

Checking feasibility of $\ILConstraints(v)$ for a state $v$ is an $\IFDL{\Size$, $2^\Size \cdot 2^{\Size'} \cdot \Size$, $\Size}$-problem that satisfies the assumptions of Lemma~\ref{lemma: IFDL} for the first case and satisfies the assumptions of Lemma~\ref{lemma: IFDL2} for the remaining two cases, and hence can be solved in (at most) exponential time in~$\Size$. Thus, checking whether a rule is applicable and applying a rule can be done in time $O(2^{g(\Size)})$, where $g(\Size)$ is a polynomial of $\Size$. 

Choosing a node to expand can be done in polynomial time in the size of the graph. As each node is re-expanded at most once (for converse compatibility), we conclude that the graph $G$ can be constructed in (at most) exponential time in~$\Size$ in the considered cases.
\myEnd
\end{proof}

\subsection{Soundness}

\begin{lemma} \label{lemma: UYDHW}
Let $G = \tuple{V,E,\nu}$ be a \CSHIQ-tableau for $\tuple{\mR,\mT,\mA}$. Then, for every simple node $v \in V$, $\FullLabel(v)$ is equivalent to $\Label(v)$. That is, for any interpretation $\mI$, $(\FullLabel(v))^\mI = (\Label(v))^\mI$.
\end{lemma}

The proof of this lemma is straightforward. 

Let $G$ be a \CSHIQ-tableau for $\tuple{\mR,\mT,\mA}$. 
For each node $v$ of $G$ with $\Status(v) \in \{\Incomplete$, $\Unsat$, $\Sat\}$, let $\DSTimeStamp(v)$ be the moment at which $\Status(v)$ was changed to its final value (i.e., determined to be $\Incomplete$, $\Unsat$ or $\Sat$). $\DSTimeStamp$ stands for ``determined-status time-stamp''. 
For each non-state $v$ of $G$, let $\ETimeStamp(v)$ be the moment at which $v$ was expanded the last time.\footnote{Each non-state may be re-expanded at most once (for making converse compatibility) and each state is expanded at most once.}

For a simple non-state $v$ with $\AfterTrans(v) = \True$, $\StatePred(v) = u$ and $\CELabelI(v) = \alpha$, we define $\CEFullLabelR(v)$ to be $\CELabelR(v)$ extended with all $\lnot S$ such that:
\begin{itemize}
\item $S \notin \CELabelR(v)$ and $S \sqsubseteq_\mR R$ for some $R \in \CELabelR(v)$, 
\item some node $w$ in the local graph of $v$ has $\Label(w)$ containing $\leq\!m\,R^-.C$ as well as $\E S^-.D$ or $\geq\!n\,S^-.D$ with $n > 0$,
\item $\alpha\!:\!C \in \FullLabel(u)$ and $\alpha\!:\!\ovl{D} \notin \FullLabel(u)$. 
\end{itemize}

For $X = \{\alpha\!:\!C_1,\ldots,\alpha\!:\!C_n\}$, let $\Cnj(X) = \alpha\!:\!(C_1 \mand \ldots \mand C_n)$. 

\begin{lemma} \label{lemma: SHQWD}
Let $G = \tuple{V,E,\nu}$ be a \CSHIQ-tableau for $\tuple{\mR,\mT,\mA}$. For every $v \in V :$
\begin{enumerate}
\item if $\Status(v) = \Unsat$ then
  \begin{enumerate}
  \item case $\Type(v) = \State$ and $\SType(v) = \Simple :$ for any predecessor $u$ of $v$ and for $u_0 = \StatePred(u)$, $u_1 = \AfterTransPred(u)$, we have that:
     \begin{enumerate}
     \item\label{ass: HGDSO} if $\Status(u_0) \neq \Incomplete$ and $\SType(u_0) = \Simple$ then there do not exist any model $\mI$ of both $\mR$ and $\mT$ and any elements $x,y \in \Delta^\mI$ such that $x \in (\FullLabel(u_0))^\mI$, $y \in (\Label(v))^\mI$, $\tuple{x,y} \in R^\mI$ for all $R \in \CELabelR(u_1)$, and $\tuple{x,y} \notin R^\mI$ for all $(\lnot R) \in \CEFullLabelR(u_1)$,
     \item\label{ass: HADSO} if $\Status(u_0) \neq \Incomplete$ and $\SType(u_0) = \Complex$ then there do not exist any model $\mI$ of $\tuple{\mR,\mT,\FullLabel(u_0)}$ and any element $y \in \Delta^\mI$ such that: $y \in (\Label(v))^\mI$, and for $x = (\CELabelI(u_1))^\mI$, $\tuple{x,y} \in R^\mI$ for all $R \in \CELabelR(u_1)$, and $\tuple{x,y} \notin R^\mI$ for all $(\lnot R) \in \CEFullLabelR(u_1)$, 
     \end{enumerate}
  \item\label{ass: JHDSM} case $\Type(v) = \State$ and $\SType(v) = \Complex :$ $\FullLabel(v)$ is unsatisfiable w.r.t.~$\mR$ and~$\mT$,
  \item\label{ass: JHWAQ} case $\Type(v) = \NonState$ and $\StatePred(v) = \Null :$ $\FullLabel(v)$ is unsatisfiable w.r.t.~$\mR$ and~$\mT$,

  \item\label{ass: UIEAL} case $\Type(v) = \NonState$, $u = \StatePred(v) \neq \Null$ and $\SType(u) = \Simple :$ if $v_0 = \AfterTransPred(v)$ then there do not exist any model $\mI$ of both $\mR$ and $\mT$ and any elements $x,y \in \Delta^\mI$ such that $x \in (\FullLabel(u))^\mI$, $y \in (\Label(v))^\mI$, $\tuple{x,y} \in R^\mI$ for all $R \in \CELabelR(v_0)$, and $\tuple{x,y} \notin R^\mI$ for all $(\lnot R) \in \CEFullLabelR(v_0)$,

  \item\label{ass: KSRLW} case $\Type(v) = \NonState$, $u = \StatePred(v) \neq \Null$ and $\SType(u) = \Complex :$ if $v_0 = \AfterTransPred(v)$ then there do not exist any model $\mI$ of $\tuple{\mR,\mT,\FullLabel(u)}$ and any element $y \in \Delta^\mI$ such that: $y \in (\Label(v))^\mI$, and for $x = (\CELabelI(v_0))^\mI$, $\tuple{x,y} \in R^\mI$ for all $R \in \CELabelR(v_0)$, and $\tuple{x,y} \notin R^\mI$ for all $(\lnot R) \in \CEFullLabelR(v_0)$,
  \end{enumerate}

\item if $\Status(v) = \Incomplete$, $\Type(v) = \State$ and $\FmlsRC(v) \neq \emptyset$ then $\FullLabel(v) \cup \{\lnot\Cnj(\FmlsRC(v))\}$ is unsatisfiable w.r.t.~$\mR$ and~$\mT$,

\item if $\Type(v) = \NonState$ and $w_1,\ldots,w_k$ are all the successors of $v$ then, for every model $\mI$ of $\mR$ and every $x \in \Delta^\mI$, 
  \begin{enumerate}
  \item case $\SType(v) = \Simple :$ $x \in (\FullLabel(v))^\mI$ iff there exists $1 \leq i \leq k$ such that $x \in (\FullLabel(w_i))^\mI$,
  \item case $\SType(v) = \Complex :$ $\mI$ is a model of $\FullLabel(v)$ iff there exists $1 \leq i \leq k$ such that $\mI$ is a model of $\FullLabel(w_i)$.
  \end{enumerate}
\end{enumerate}
\end{lemma}

\begin{proof} We prove this lemma by induction on both $\DSTimeStamp(v)$ and $\ETimeStamp(v)$. 

Consider the case~\ref{ass: HGDSO} when $v$ gets status $\Unsat$ because $\ILConstraints(v)$ is infeasible. For this case, we prove the contrapositive. Suppose that:
\begin{equation}\label{eq: YUENB} 
\parbox{14cm}{$\Type(v) = \State$, $\SType(v) = \Simple$, $\tuple{u,v} \in E$, $u_0 = \StatePred(u)$, $u_1 = \AfterTransPred(u)$, $\Status(u_0) \neq \Incomplete$, $\SType(u_0) = \Simple$, $\mI$ is a model of both $\mR$ and $\mT$, $x \in (\FullLabel(u_0))^\mI$, $y \in (\Label(v))^\mI$, $\tuple{x,y} \in R^\mI$ for all $R \in \CELabelR(u_1)$, and $\tuple{x,y} \notin R^\mI$ for all $(\lnot R) \in \CEFullLabelR(u_1)$.}
\end{equation}
We show that $\ILConstraints(v)$ is feasible.
Without loss of generality, assume that $\mI$ is finitely-branching.\footnote{It is known that the DL \SHIQ has the finitely-branching model property.}
Thus, the set $Z = \{z \in \Delta^\mI \mid z \neq x$, $\tuple{y,z} \in R^\mI$ for some $R \in \RN\cup\RN^-\}$ is finite. 
Let us compute a solution $\mS$ for $\ILConstraints(v)$ as follows.
\begin{enumerate}
\item For each successor $w$ of $v$, set $n_w := 0$.
\item For each $z \in Z$ do:
  \begin{enumerate}
  \item\label{item: GHSDO} let $w_1,\ldots,w_k$ be all the successors of $v$ such that, for each $1 \leq i \leq k\,$:
     \begin{enumerate}
     \item\label{item: JDJAP} $\CELabelT(w_i) = \CQF$, 
     \item\label{item: JDYTS} $z \in (Label(w_i))^\mI$,
     \item\label{item: JPWBP} $\tuple{y,z} \in R^\mI$ for all $R \in \CELabelR(w_i)$,
     \item $w_i$ is ``maximal'' in the sense that there does not exist any successor $w'_i \neq w_i$ of $v$ such that
	\begin{itemize}
	\item $\Label(w'_i) \supseteq \Label(w_i)$ and $\CELabelR(w'_i) \supseteq \CELabelR(w_i)$, 
	\item $\CELabelT(w'_i) = \CQF$, 
	\item $z \in (\Label(w'_i))^\mI$,
	\item $\tuple{y,z} \in R^\mI$ for all $R \in \CELabelR(w'_i)$;
	\end{itemize}
     \end{enumerate} 

  \item\label{item: JRIAH} for each $1 \leq i \leq k$, set $n_{w_i} := n_{w_i} + 1$.
  \end{enumerate}
\item $\mS_0 := \{x_w = n_w \mid w$ is a successor of $v\}$.
\item\label{item: HWUZS} Let $w_1,\ldots,w_k$ be all the successors of $v$ defined as in the step~\ref{item: GHSDO} for the case $z = x$.
\item For each $1 \leq i \leq k$ do: 
  \begin{enumerate}
  \item if there exists $\succeq\!n\,S.D \in \Label(v)$ such that $S \in \CELabelR(w_i)$, $D \in \Label(w_i)$, $S^- \notin \CELabelR(u_1)$ and $\ovl{D} \notin \FullLabel(u_0)$ then 
	\begin{enumerate} 
	\item\label{item: HGRES} set $n_{w_i} := n_{w_i} + 1$.
	\end{enumerate}
  \end{enumerate}
\item $\mS := \{x_w = n_w \mid w$ is a successor of $v\}$.
\end{enumerate}

We prove that $\mS$ is a solution for $\ILConstraints(v)$. 

We first show that, for any $w_i$ at the step~\ref{item: GHSDO} or~\ref{item: HWUZS}, $\tuple{y,z} \notin S^\mI$ for all $(\lnot S) \in \CEFullLabelR(w_i)$. 
Let $(\lnot S) \in \CEFullLabelR(w_i)$. Thus, there exist $R$, $w'$, $m$, $C$, $n$, $D$ such that: $S \sqsubseteq_\mR R$, $R \in \CELabelR(w_i)$, $w'$ is a node in the local graph of $w_i$, $\Label(w')$ contains $\leq\!m\,R^-.C$ as well as $\E S^-.D$ or $\geq\!n\,S^-.D$ with $n > 0$, $C \in \FullLabel(v)$ and $\ovl{D} \notin \FullLabel(v)$. As $\Status(v) \neq \Incomplete$ and the rule $\KCCcp$ was not applicable to $w'$, either $\E S.\top_R$ or $\V S.\bot$ must belong to $\Label(v)$. If $\V S.\bot \in \Label(v)$ then, since $y \in (\Label(v))^\mI$, we have that $y \in (\V S.\bot)^\mI$ and therefore $\tuple{y,z} \notin S^\mI$. Suppose $\E S.\top_R \in \Label(v)$. If $\tuple{y,z} \in S^\mI$ then, by the nature of the transitional full-expansion rule and the maximality of $w_i$, we have that $\top_R \in \Label(w_i)$ and $S \in \CELabelR(w_i)$, which contradicts the assumption $(\lnot S) \in \CEFullLabelR(w_i)$. Therefore, $\tuple{y,z} \notin S^\mI$. 

We now show that if a constraint $x_w = 0$ was added to $\ILConstraints(v)$ because $w$ got status $\Unsat$ then $n_w$ was not increased and hence must be 0. For the contrary, suppose the constraint $x_w = 0$ was added to $\ILConstraints(v)$ (because $w$ got status $\Unsat$) and $n_w$ was increased at least once. Thus, there exists $z \in Z \cup \{x\}$ such that $z \in (\Label(w))^\mI$ and $\tuple{y,z} \in R^\mI$ for all $R \in \CELabelR(w)$. By the assertion stated in the above paragraph, we also have that $\tuple{y,z} \notin S^\mI$ for all $(\lnot S) \in \CEFullLabelR(w)$. By~\eqref{eq: YUENB}, we have that $y \in (\Label(v))^\mI$ and $\SType(v) = \Simple$, and by Lemma~\ref{lemma: UYDHW}, it follows that $y \in (\FullLabel(v))^\mI$. This situation contradicts the inductive assumption~\ref{ass: UIEAL} (with $v$, $v_0$, $u$, $x$, $y$ replaced by $w$, $w$, $v$, $y$, $z$, respectively). Therefore, every constraint $x_w = 0$ from $\ILConstraints(v)$ is satisfied by the solution~$\mS$. 

Consider a concept $\succeq\!n\,S.D \in \Label(v)$ and the corresponding constraint $\sum \{ x_w \mid S \in \CELabelR(w)$, $D \in \Label(w)\} \geq n$ of $\ILConstraints(v)$. There are the following cases:
\begin{itemize}
\item Case $\geq\!(n+1)\,S.D \in \Label(v)$, $S^- \in \CELabelR(u_1)$ and $D \in \FullLabel(u_0)$: By~\eqref{eq: YUENB}, we have that $y \in (\geq\!(n+1)\,S.D)^\mI$, $\tuple{y,x} \in S^\mI$ and $x \in D^\mI$. Hence, $Z$ contains pairwise different $z_1, \ldots, z_n$ such that $\tuple{y,z_i} \in S^\mI$ and $z_i \in D^\mI$ for all $1 \leq i \leq n$. Each $z_i$ makes $n_w$ increased by~1 for some successor $w$ of $v$ with $S \in \CELabelR(w)$ and $D \in \Label(w)$. It follows that the considered constraint is satisfied by the solution~$\mS$.
\item Case ($\geq\!n\,S.D \in \Label(v)$ or ($\E S.D \in \Label(v)$ and $n=1$)) and $\ovl{D} \in \FullLabel(u_0)$: By~\eqref{eq: YUENB}, we have that $y \in (\geq\!n\,S.D)^\mI$ and $x \notin D^\mI$. Analogously to the above case, the considered constraint is satisfied by the solution~$\mS$.
\item Case ($\geq\!n\,S.D \in \Label(v)$ or ($\E S.D \in \Label(v)$ and $n=1$)), $\ovl{D} \notin \FullLabel(u_0)$ and $S^- \in \CELabelR(u_1)$: By the rule \KCCc, $D \in \FullLabel(u_0)$, which implies that \mbox{$\geq\!(n+1)\,S.D \in \Label(v)$}. Thus, this case is reduced to the first one.  
\item Case ($\geq\!n\,S.D \in \Label(v)$ or ($\E S.D \in \Label(v)$ and $n=1$)), $\ovl{D} \notin \FullLabel(u_0)$ and $S^- \notin \CELabelR(u_1)$: By~\eqref{eq: YUENB}, we have that $y \in (\geq\!n\,S.D)^\mI$. If $\tuple{y,x} \notin S^\mI$ or $x \notin D^\mI$ then, analogously to the first case, the considered constraint is satisfied by the solution~$\mS$. Assume that $\tuple{y,x} \in S^\mI$ and $x \in D^\mI$. We have that $Z \cup \{x\}$ contains pairwise different $z_1 = x, z_2, \ldots, z_n$ such that $\tuple{y,z_i} \in S^\mI$ and $z_i \in D^\mI$, for $1 \leq i \leq n$. Each $z_i$ makes $n_w$ increased by~1 for some successor $w$ of $v$ with $S \in \CELabelR(w)$ and $D \in \Label(w)$. It follows that the considered constraint is satisfied by the solution~$\mS$.
\end{itemize}

Consider a concept $\preceq\!m\,R.C \in \Label(v)$ and the corresponding constraint $\sum \{ x_w \mid R \in \CELabelR(w)$, $C \in \Label(w)\} \leq m$ of $\ILConstraints(v)$. There are two cases: 
\begin{eqnarray}
& \mbox{- case} & (\leq\!m\,R.C) \in \Label(v),\label{eq: JSTAO}\\
& \mbox{- case} & (\leq\!(m+1)\,R.C) \in \Label(v), R^- \in \CELabelR(u_1) \mbox{ and } C \in \FullLabel(u_0).\label{eq: HGELS}
\end{eqnarray}

Consider the first case, i.e., assume that~\eqref{eq: JSTAO} holds. 
By~\eqref{eq: YUENB}, we have that $y \in (\leq\!m\,R.C)^\mI$. Hence, $Z \cup \{x\}$ contains no more than $m$ elements $z$ such that $\tuple{y,z} \in R^\mI$ and $z \in C^\mI$. Due to the ``maximality'' of $w$ and the nature of the transitional full-expansion rule, for such a~$z$ there exists at most one successor $w$ of $v$ such that $R \in \CELabelR(w)$, $C \in \Label(w)$ and the consideration of $z$ causes $n_w$ to be increased by~1. (Also, such an $n_w$ is increased only due to such a $z$.) Therefore, the considered constraint is satisfied by the solution~$\mS$. 

Consider the second case, i.e., assume that~\eqref{eq: HGELS} holds. 
By~\eqref{eq: YUENB}, we have that \mbox{$y \in (\leq\!(m+1)\,R.C)^\mI$}, $\tuple{y,x} \in R^\mI$ and $x \in C^\mI$. 
Hence, $Z$ contains no more than $m$ elements $z$ such that $\tuple{y,z} \in R^\mI$ and $z \in C^\mI$. Due to the ``maximality'' of $w$ and the nature of the transitional full-expansion rule, for such a~$z$ there exists at most one successor $w$ of $v$ such that $R \in \CELabelR(w)$, $C \in \Label(w)$ and the consideration of $z$ causes $n_w$ to be increased by~1. (Also, such an $n_w$ is increased only due to such a $z$.) Therefore, the considered constraint is satisfied by~$\mS_0$. To prove that it is also satisfied by~$\mS$, it suffices to show that if $n_{w_i}$ was increased at the step~\ref{item: HGRES} (in the construction of~$\mS$) then $R \notin \CELabelR(w_i)$ or $C \notin \Label(w_i)$. 

For the contrary, suppose that $n_{w_i}$ was increased at the step~\ref{item: HGRES} and 
\begin{equation}\label{eq: HDKAJ} 
R \in \CELabelR(w_i) \mbox{ and } C \in \Label(w_i).
\end{equation} 
As the condition of the step~\ref{item: HGRES}, 
\begin{equation}\label{eq: JHRIA}
\parbox{14cm}{there exists $\succeq\!n\,S.D \in \Label(v)$ such that $S \in \CELabelR(w_i)$, $D \in \Label(w_i)$, $S^- \notin \CELabelR(u_1)$, $\ovl{D} \notin \FullLabel(u_0)$, and hence $\tuple{y,x} \in S^\mI$ and $x \in D^\mI$.}
\end{equation}

Since $S^- \notin \CELabelR(u_1)$ (by~\eqref{eq: JHRIA}) and $R^- \in \CELabelR(u_1)$ (by~\eqref{eq: HGELS}), we have that $R \not\sqsubseteq_\mR S$. 

Consider the case $S \not\sqsubseteq_\mR R$. 
Since both $S$ and $R$ belong to $\CELabelR(w_i)$ (by~\eqref{eq: JHRIA} and~\eqref{eq: HDKAJ}), there exist roles
\begin{equation}\label{eq: HGDEO}
\parbox{14cm}{$R_0 = R, R_1, \ldots, R_{h-1}, R_h = S$ and $S_1, \ldots, S_h$, all belonging to $\CELabelR(w_i)$}
\end{equation}
such that, for every $1 \leq j \leq h$:
\begin{eqnarray}
& - & S_j \sqsubseteq_\mR R_{j-1} \mbox{ and } S_j \sqsubseteq_\mR R_j,\label{eq: HDKJS}\\
& - & \mbox{$\Label(v)$ contains $\E S_j.D'_j$ or $\geq\!n_j\,S_j.D'_j$ for some $D'_j \in \Label(w_i)$ and $n_j > 0$},\label{eq: JROSH}\\
& - & \mbox{if $j < h$ then $\Label(v)$ contains $\leq\!m_j\,R_j.C'_j$ for some $C'_j \in \Label(w_i)$ and $m_j$.}\label{eq: JEOSA}
\end{eqnarray}
For each $j$ from 1 to $h$, observe that:
\begin{itemize}
\item since $S_j \in \CELabelR(w_i)$, by the step~\ref{item: HWUZS} and the step~\ref{item: JPWBP} (with $z = x$) in the construction of the solution~$\mS$, we have that 
\begin{equation}\label{eq: HGDAK}
\tuple{y,x} \in S_j^\mI
\end{equation}
\item if $S_j^- \notin \CELabelR(u_1)$ then 
   \begin{itemize}
   \item $(\lnot S_j^-) \in \CEFullLabelR(u_1)$ because one can apply the definition of $\CEFullLabelR$ (with $u$, $v$, $w$, $\alpha$, $S$, $R$, $m$, $C$, $n$, $D$ replaced, respectively, by $u_0$, $u_1$, $u$, $\Null$, $S_j^-$, $R_{j-1}^-$, $m_{j-1}$, $C'_{j-1}$, $n_j$, $D'_j$) due to the following reasons:
	\begin{itemize}
	\item by~\eqref{eq: YUENB}, $\AfterTrans(u_1) = \True$ and $\StatePred(u_1) = u_0$,  
	\item $S_j^- \notin \CELabelR(u_1)$ (by the premise of the above ``if'' clause), 
	\item $S_j^- \sqsubseteq R_{j-1}^-$ (by~\eqref{eq: HDKJS})
	\item $R_{j-1}^- \in \CELabelR(u_1)$ since
	   \begin{itemize}
	   \item if $j=1$ then $R_{j-1}^- = R_0^- = R^-$ (by \eqref{eq: HGDEO}) and $R^- \in \CELabelR(u_1)$ (by~\eqref{eq: HGELS}),
	   \item if $j>1$ then $R_{j-1}^- \in \CELabelR(u_1)$ by induction of~\eqref{eq: KJERA} as shown below, 
	   \end{itemize}
	\item by~\eqref{eq: YUENB}, $u$ is a node in the local graph of $u_1$ and, since $\tuple{u,v} \in E$ and $v$ is a state, like $\Label(v)$, $\Label(u)$ contains $\leq\!m_{j-1}\,R_{j-1}.C'_{j-1}$ (by~\eqref{eq: JEOSA} and~\eqref{eq: HGELS}, using $C'_0 = C$ and $m_0 = m+1$) as well as $\E S_j.D'_j$ or $\geq\!n_j\,S_j.D'_j$ with $n_j > 0$ (by~\eqref{eq: JROSH}), 
	\item since $R_{j-1}^- \in \CELabelR(u_1)$ (as above), $(\leq\!m_{j-1}\,R_{j-1}.C'_{j-1}) \in \Label(u)$ (as above) and $\Status(u_0) \neq \Incomplete$ (by~\eqref{eq: YUENB}), by the rule \KCCc, either $C'_{j-1}$ or $\ovl{C'_{j-1}}$ belongs to $\FullLabel(u_0)$,
	\item if $j > 1$ then, since $C'_{j-1} \in \Label(w_i)$ (by~\eqref{eq: JEOSA}), by the step~\ref{item: HWUZS} and the step~\ref{item: JDYTS} (with $z = x$) in the construction of the solution~$\mS$, we have that $x \in (C'_{j-1})^\mI$, hence, by~\eqref{eq: YUENB}, $\ovl{C'_{j-1}} \notin \FullLabel(u_0)$, and by the assertion in the above item, $C'_{j-1} \in \FullLabel(u_0)$, 
	\item if $j = 1$ then, since $C \in \FullLabel(u_0)$ (by~\eqref{eq: HGELS}), for $C'_0 = C$ we have that $C'_{j-1} \in \FullLabel(u_0)$, 
	\item since $D'_j \in \Label(w_i)$ (by~\eqref{eq: JROSH}), by the step~\ref{item: HWUZS} and the step~\ref{item: JDYTS} (with $z = x$) in the construction of the solution~$\mS$, we have that $x \in (D'_j)^\mI$, hence, by~\eqref{eq: YUENB}, $\ovl{D'_j} \notin \FullLabel(u_0)$;
	\end{itemize}
   \item since $(\lnot S_j^-) \in \CEFullLabelR(u_1)$ (shown above), by~\eqref{eq: YUENB}, $\tuple{y,x} \notin S_j^\mI$, which contradicts~\eqref{eq: HGDAK};
   \end{itemize}
\item hence $S_j^- \in \CELabelR(u_1)$, and by~\eqref{eq: HDKJS}, it follows that
\begin{equation}\label{eq: KJERA}
R_j^- \in \CELabelR(u_1)
\end{equation}
\end{itemize}
As a consequence, for $j = h$, we have that $S_h^- \in \CELabelR(u_1)$. Since $S_h \sqsubseteq_\mR R_h$ (by~\eqref{eq: HDKJS}) and $R_h = S$ (by~\eqref{eq: HGDEO}), it follows that $S^- \in \CELabelR(u_1)$, which contradicts~\eqref{eq: JHRIA}. 

Now consider the case $S \sqsubseteq_\mR R$. One can derive $(\lnot S^-) \in \CEFullLabelR(u_1)$ by applying the definition of $\CEFullLabelR$ (with $u$, $v$, $w$, $\alpha$, $S$, $R$, $m$, $n$ replaced, respectively, by $u_0$, $u_1$, $u$, $\Null$, $S^-$, $R^-$, $m+1$, and some $n'$) due to the following reasons:
\begin{itemize}
\item by~\eqref{eq: YUENB}, $\AfterTrans(u_1) = \True$ and $\StatePred(u_1) = u_0$,  
\item $S^- \notin \CELabelR(u_1)$ (by~\eqref{eq: JHRIA}), $S^- \sqsubseteq_\mR R^-$ (since $S \sqsubseteq_\mR R$), $R^- \in \CELabelR(u_1)$ (by~\eqref{eq: HGELS}), 
\item by~\eqref{eq: YUENB}, $u$ is a node in the local graph of $u_1$ and, since $\tuple{u,v} \in E$ and $v$ is a state, like $\Label(v)$, $\Label(u)$ contains $(\leq\!(m+1)\,R.C)$ (by~\eqref{eq: HGELS}) as well as $\E S.D$ or $\geq\!n'\,S.D$ for some $n' > 0$ (since $\succeq\!n\,S.D \in \Label(v)$ -- by~\eqref{eq: JHRIA}),
\item $C \in \FullLabel(u_0)$ (by~\eqref{eq: HGELS}) and $\ovl{D} \notin \FullLabel(u_0)$ (by~\eqref{eq: JHRIA}).
\end{itemize}
By~\eqref{eq: YUENB}, it follows that $\tuple{y,x} \notin S^\mI$, which contradicts~\eqref{eq: JHRIA}.

We have proved the induction step for the case~\ref{ass: HGDSO} when $v$ gets status $\Unsat$ because $\ILConstraints(v)$ is infeasible. The case~\ref{ass: HADSO} when $v$ gets status $\Unsat$ because $\ILConstraints(v)$ is infeasible can be dealt with in a similar way, using the following modifications, with $a = \CELabelI(u_1)$ (and thus $x = a^\mI$):
\begin{itemize}
\item The assumption~\eqref{eq: YUENB} is modified by changing ``$\SType(u_0) = \Simple$'' to ``$\SType(u_0) = \Complex$'' and changing ``$\mI$ is a model of both $\mR$ and $\mT$, $x \in (\FullLabel(u_0))^\mI\,$'' to ``$\mI$ is a model of $\tuple{\mR,\mT,\FullLabel(u_0)}$''.

\item ``$D \in \FullLabel(u_0)$'' is replaced by ``$a\!:\!D \in \FullLabel(u_0)$'', and similarly for other cases with another concept in the place of $D$ and/or with $\notin$ instead of $\in$.  

\item The phrase ``either $C'_{j-1}$ or $\ovl{C'_{j-1}}$ belongs to $\FullLabel(u_0)$'' is replaced by ``either $a\!:\!C'_{j-1}$ or $a\!:\!\ovl{C'_{j-1}}$ belongs to $\FullLabel(u_0)$''.

\item The phrase ``$u$, $\Null$'' is replaced by ``$u$, $a$''. 
\end{itemize}


Now consider the case~\ref{ass: JHDSM} when $v$ gets status $\Unsat$ because $\ILConstraints(v)$ is infeasible. For this case, we prove the contrapositive. Suppose that $\Type(v) = \State$, $\SType(v) = \Complex$ and $\FullLabel(v)$ is satisfiable w.r.t.\ $\mR$ and $\mT$. We show that $\ILConstraints(v)$ is feasible. 
Let $\mI$ be a finitely-branching model of $\mR$, $\mT$ and $\FullLabel(v)$. 
We compute a solution $\mS$ for $\ILConstraints(v)$ as follows.
\begin{enumerate}
\item For each successor $w$ of $v$, set $n_w := 0$.
\item For each individual $a$ occurring in $\mA$ and each $z \in \Delta^\mI$ such that $\tuple{a^\mI,z} \in R^\mI$ for some $R \in \RN \cup \RN^-$ do:
  \begin{enumerate}
  \item\label{item: XHSDO} let $w_1,\ldots,w_k$ be all the successors of $v$ such that, for each $1 \leq i \leq k\,$:
     \begin{enumerate}
     \item\label{item: XDJAP} $\CELabelT(w_i) = \CQF$ and $\CELabelI(w_i) = a$, 
     \item\label{item: XDYTS} $z \in (Label(w_i))^\mI$,
     \item\label{item: XPWBP} $\tuple{a^\mI,z} \in R^\mI$ for all $R \in \CELabelR(w_i)$,
     \item $w_i$ is ``maximal'' in the sense that there does not exist any successor $w'_i \neq w_i$ of $v$ such that
	\begin{itemize}
	\item $\Label(w'_i) \supseteq \Label(w_i)$ and $\CELabelR(w'_i) \supseteq \CELabelR(w_i)$, 
	\item $\CELabelT(w'_i) = \CQF$ and $\CELabelI(w'_i) = a$, 
	\item $z \in (\Label(w'_i))^\mI$,
	\item $\tuple{a^\mI,z} \in R^\mI$ for all $R \in \CELabelR(w'_i)$;
	\end{itemize}
     \end{enumerate} 

  \item\label{item: XRIAH} for each $1 \leq i \leq k$ do
     \begin{enumerate}
     \item\label{item: XRIYH} if $z \neq b^\mI$ for all $b$ occurring in $\mA$ then $n_{w_i} := n_{w_i} + 1$;
     \item\label{item: XRIXH} else if $z = b^\mI$ for some $b$ occurring in $\mA$ and there exists $a\!:\!(\succeq\!l\,S.D) \in \Label(v)$ such that $S \in \CELabelR(w_i)$, $D \in \Label(w_i)$ and $S(a,b) \notin \Label(v)$ then $n_{w_i} := n_{w_i} + 1$.
     \end{enumerate}
  \end{enumerate}
\item $\mS := \{x_w = n_w \mid w$ is a successor of $v\}$.
\end{enumerate}

We prove that $\mS$ is a solution for $\ILConstraints(v)$. 

We first show that, for any $w_i$ at the step~\ref{item: XHSDO}, $\tuple{a^\mI,z} \notin S^\mI$ for all $(\lnot S) \in \CEFullLabelR(w_i)$. 
Let $(\lnot S) \in \CEFullLabelR(w_i)$. Thus, there exist $R$, $w'$, $m$, $C$, $n$, $D$ such that: $S \sqsubseteq_\mR R$, $R \in \CELabelR(w_i)$, $w'$ is a node in the local graph of $w_i$, $\Label(w')$ contains $\leq\!m\,R^-.C$ as well as $\E S^-.D$ or $\geq\!n\,S^-.D$ with $n > 0$, $a\!:\!C \in \FullLabel(v)$ and $a\!:\!\ovl{D} \notin \FullLabel(v)$. As $\Status(v) \neq \Incomplete$ and the rule $\KCCcp$ was not applicable to $w'$, either $a\!:\!\E S.\top_R$ or $a\!:\!\V S.\bot$ must belong to $\Label(v)$. If $(a\!:\!\V S.\bot) \in \Label(v)$ then $a^\mI \in (\V S.\bot)^\mI$ and therefore $\tuple{a^\mI,z} \notin S^\mI$. Suppose $(a\!:\!\E S.\top_R) \in \Label(v)$. If $\tuple{a^\mI,z} \in S^\mI$ then, by the nature of the transitional full-expansion rule and the maximality of $w_i$, we have that $\top_R \in \Label(w_i)$ and $S \in \CELabelR(w_i)$, which contradicts the assumption $(\lnot S) \in \CEFullLabelR(w_i)$. Therefore, $\tuple{a^\mI,z} \notin S^\mI$. 

We now show that if a constraint $x_w = 0$ was added to $\ILConstraints(v)$ because $w$ got status $\Unsat$ then $n_w$ was not increased and hence must be 0. For the contrary, suppose the constraint $x_w = 0$ was added to $\ILConstraints(v)$ (because $w$ got status $\Unsat$) and $n_w$ was increased at least once. Thus, there exists $z \in \Delta^\mI$ such that $z \in (\Label(w))^\mI$ and $\tuple{a^\mI,z} \in R^\mI$ for all $R \in \CELabelR(w)$. By the assertion stated in the above paragraph, we also have that $\tuple{a^\mI,z} \notin S^\mI$ for all $(\lnot S) \in \CEFullLabelR(w)$. This situation contradicts the inductive assumption~\ref{ass: KSRLW} (with $v$, $v_0$, $u$, $x$, $y$ replaced by $w$, $w$, $v$, $a^\mI$, $z$, respectively). Therefore, every constraint $x_w = 0$ from $\ILConstraints(v)$ is satisfied by the solution~$\mS$. 

Consider a concept $(a\!:\,\succeq\!n\,S.D) \in \Label(v)$ and the corresponding constraint $\sum \{ x_w \mid \CELabelI(w) = a$, $S \in \CELabelR(w)$, $D \in \Label(w)\} \geq n$ of $\ILConstraints(v)$. Let $m = \sharp\{b \mid \{S(a,b), b\!:\!D\} \subseteq \FullLabel(v)\}$. We have that $(a\!:\,\geq\!(n+m)\,S.D) \in \Label(v)$. Since $\mI$ is a model of $\FullLabel(v)$, there exist pairwise different $z_1,\ldots,z_{n+m}$ such that $\tuple{a^\mI,z_i} \in S^\mI$ and $z_i \in D^\mI$ for all $1 \leq i \leq n+m$. Note that, if $S(a,b) \in \Label(v)$ then, by the rule $\KCCc$, either $b\!:\!D \in \FullLabel(v)$ or $b\!:\!\ovl{D} \in \FullLabel(v)$. Since $z_i \in D^\mI$ and $\mI$ is a model of $\FullLabel(v)$, if $z_i = b^\mI$ then $b\!:\!\ovl{D} \notin \FullLabel(v)$. Therefore, for every $1 \leq i \leq n+m$, if $z_i = b^\mI$ and $S(a,b) \in \Label(v)$ then $b\!:\!D \in \FullLabel(v)$. Let $Z = \{z_1,\ldots,z_{n+m}\} \setminus \{b^\mI \mid S(a,b) \in \Label(v)\}$. We have that $\sharp Z = n$. Each $z$ from $Z$ makes $n_w$ increased by~1 for some successor $w$ of $v$ with $\CELabelI(w) = a$, $S \in \CELabelR(w)$ and $D \in \Label(w)$. It follows that the considered constraint is satisfied by the solution~$\mS$.

Consider a concept $(a\!:\,\preceq\!n\,R.C) \in \Label(v)$ and the corresponding constraint $\sum \{ x_w \mid \CELabelI(w) = a$, $R \in \CELabelR(w)$, $C \in \Label(w)\} \leq n$ of $\ILConstraints(v)$. Let $m = \sharp\{b \mid \{R(a,b), b\!:\!C\} \subseteq \FullLabel(v)\}$. We have that $(a\!:\,\leq\!(n+m)\,R.C) \in \Label(v)$. Since $\mI$ is a model of $\FullLabel(v)$, it follows that $a^\mI \in (\leq\!(n+m)\,R.C)^\mI$. Let $Z_1 = \{b^\mI \mid \{R(a,b), b\!:\!C\} \subseteq \FullLabel(v)\}$. Due to the subrule~\ref{item: JHEAA} of~\NUS, we have that $\sharp Z_1 = m$. 

Note that if $w$ is a successor of $v$, $\CELabelI(w) = a$, $R \in \CELabelR(w)$ and $C \in \Label(w)$ then $n_w$ is increased only due to some $z$ such that $\tuple{a^\mI,z} \in R^\mI$ and $z \in C^\mI$. 
Due to the ``maximality'' of $w$ and the nature of the transitional full-expansion rule, for such a~$z$ there exists at most one successor $w$ of $v$ such that $\CELabelI(w) = a$, $R \in \CELabelR(w)$, $C \in \Label(w)$ and the consideration of $z$ causes $n_w$ to be increased by~1. Since $a^\mI \in (\leq\!(n+m)\,R.C)^\mI$, to prove that the considered constraint is satisfied by the solution~$\mS$, it suffices to show that if $z \in Z_1$ causes $n_{w_i}$ to be increased by~1 at the step~\ref{item: XRIXH} then $R \notin \CELabelR(w_i)$ or $C \notin \Label(w_i)$. Suppose the contrary. We have that:
\begin{eqnarray}
& - & \{a\!:\,\leq\!(n+m)\,R.C, R(a,b), b\!:\!C\} \subseteq \FullLabel(v),\label{eq: ODJSA}\\
& - & w_i \mbox{ is a successor of } v, \CELabelT(w_i) = \CQF \mbox{ and } \CELabelI(w_i) = a,\label{eq: HJFGA}\\
& - & b^\mI \in (\Label(w_i))^\mI \mbox{ and } \tuple{a^\mI,b^\mI} \in (R')^\mI \mbox{ for all } R' \in \CELabelR(w_i),\label{eq: IUFDA}\\
& - & a\!:\!(\succeq\!l\,S.D) \in \Label(v), S \in \CELabelR(w_i), D \in \Label(w_i) \mbox{ and } S(a,b) \notin \Label(v),\label{eq: KLDFI}\\
& - & R \in \CELabelR(w_i) \mbox{ and } C \in \Label(w_i).\label{eq: JDFLA} 
\end{eqnarray}

Since both $S$ and $R$ belong to $\CELabelR(w_i)$ (by~\eqref{eq: KLDFI} and~\eqref{eq: JDFLA}), there exist roles
\begin{equation}\label{eq: FDKJO}
\parbox{14cm}{$R_0 = R, R_1, \ldots, R_{h-1}, R_h = S$ and $S_1, \ldots, S_h$, all belonging to $\CELabelR(w_i)$}
\end{equation}
such that, for every $1 \leq j \leq h$:
\begin{eqnarray}
& - & S_j \sqsubseteq_\mR R_{j-1} \mbox{ and } S_j \sqsubseteq_\mR R_j,\label{eq: KJDFS}\\
& - & \mbox{$\Label(v)$ contains $a\!:\!\E S_j.D'_j$ or $a\!:\,\geq\!n_j\,S_j.D'_j$ for some $D'_j \in \Label(w_i)$ and $n_j > 0$},\label{eq: JHFDS}\\
& - & \mbox{if $j < h$ then $\Label(v)$ contains $a\!:\,\leq\!m_j\,R_j.C'_j$ for some $C'_j \in \Label(w_i)$ and $m_j$.}\label{eq: HJFGS}
\end{eqnarray}

Note that the subrule~\ref{item: OSJRS} of \NUS was not applicable to $v$. Having \eqref{eq: ODJSA}, \eqref{eq: FDKJO}, \eqref{eq: KJDFS} and \eqref{eq: JHFDS}, we derive that $S_1(a,b) \in \Label(v)$ or $\lnot S_1(a,b) \in \Label(v)$.  Since~\eqref{eq: FDKJO} and~\eqref{eq: IUFDA}, $\tuple{a^\mI,b^\mI} \in S_1^\mI$. Since $\mI$ is a model of $\FullLabel(v)$, it follows that $\lnot S_1(a,b) \notin \Label(v)$, and hence $S_1(a,b) \in \Label(v)$. Since $S_1 \sqsubseteq R_1$ (by~\eqref{eq: KJDFS}), by the rule \US, we also have that $R_1(a,b) \in \Label(v)$. Analogously, using also~\eqref{eq: HJFGS}, for every $j$ from 1 to $h$, we can derive that $S_j(a,b) \in \Label(v)$ and $R_j(a,b) \in \Label(v)$. Since $S = R_h$, it follows that $S(a,b) \in \Label(v)$, which contradicts~\eqref{eq: KLDFI}. 
This completes the induction step for the case~\ref{ass: JHDSM} when $v$ gets status $\Unsat$ because $\ILConstraints(v)$ is infeasible. 


The induction steps for the other cases are straightforward. 
\myEnd
\end{proof}

\begin{corollary}[Soundness of \CSHIQ]
If $G = \tuple{V,E,\nu}$ is a \CSHIQ-tableau for $\tuple{\mR,\mT,\mA}$ and $\Status(\nu) = \Unsat$ then $\tuple{\mR,\mT,\mA}$ is unsatisfiable.
\end{corollary}

This corollary follows from the assertion~\ref{ass: JHWAQ} of Lemma~\ref{lemma: SHQWD}.


\subsection{Completeness}

%

We prove completeness of \CSHIQ via model graphs. The technique has been used for other logics (e.g., in~\cite{Rautenberg83,Gore99,nguyen01B5SL,dkns2011}).
A {\em model graph} is a tuple $\langle \Delta, \nI, \nC, \nE \rangle$, where:
\begin{itemize}
\item $\Delta$ is a non-empty set, 
\item $\nI$ is a mapping that associates each individual name with an element of $\Delta$, 
\item $\nC$ is a mapping that associates each element of $\Delta$ with a set of concepts, 
\item $\nE$ is a mapping that associates each role with a binary relation on $\Delta$.
\end{itemize}

A model graph $\langle \Delta, \nI, \nC, \nE \rangle$ is {\em consistent} and {\em $\mR$-saturated} if every $x \in \Delta$ satisfies:\footnote{A consistent and $\mR$-saturated model graph is like a Hintikka structure.} 
\begin{eqnarray}
&-\ & \textrm{$\nC(x)$ does not contain $\bot$ nor any pair $C$, $\ovl{C}$}\label{eq:HGDSX 0}\\
&-\ & \textrm{if $\tuple{x,y} \in \nE(R)$ then $\tuple{y,x} \in \nE(R^-)$}\label{eq:HGDSX a}\\
&-\ & \textrm{if $\tuple{x,y} \in \nE(R)$ and $R \sqsubseteq_\mR S$ then $\tuple{x,y} \in \nE(S)$}\label{eq:HGDSX b}\\
&-\ & \textrm{if $C \mand D \in \nC(x)$ then $\{C,D\} \subseteq \nC(x)$}\label{eq:HGDSX 1}\\
&-\ & \textrm{if $C \mor D \in \nC(x)$ then $C \in \nC(x)$ or $D \in \nC(x)$}\label{eq:HGDSX 2}\\
&-\ & \textrm{if $\V S.C \in \nC(x)$ and $R \sqsubseteq_\mR S$ then $\V R.C \in \nC(x)$}\label{eq:HGDSX 3}\\
&-\ & \textrm{if $\tuple{x,y} \in \nE(R)$ and $\V R.C \in \nC(x)$ then $C \in \nC(y)$}\label{eq:HGDSX 4}\\
&-\ & \textrm{if $\tuple{x,y} \in \nE(R)$, $\Trans{R}$ and $\V R.C \in \nC(x)$ then $\V R.C \in \nC(y)$}\label{eq:HGDSX 4b}\\
&-\ & \textrm{if $\E R.C \in \nC(x)$ then $\E y \in \Delta$ such that $\tuple{x,y} \in \nE(R)$ and $C \in \nC(y)$}\label{eq:HGDSX 6}\\
&-\ & \textrm{if $(\geq\!n\,R.C) \in \nC(x)$ then $\sharp\{\tuple{x,y} \in \nE(R) \mid C \in \nC(y)\} \geq n$}\label{eq:HGDSX 7} \\
&-\ & \textrm{if $(\leq\!n\,R.C) \in \nC(x)$ then $\sharp\{\tuple{x,y} \in \nE(R) \mid C \in \nC(y)\} \leq n$}\label{eq:HGDSX 8} \\
&-\ & \textrm{if $(\leq\!n\,R.C) \in \nC(x)$ and $\tuple{x,y}\in \nE(R)$ then $C \in \nC(y)$ or $\ovl{C} \in \nC(y)$.}\label{eq:HGDSX 9}
\end{eqnarray}

Given a~model graph $M = \langle \Delta, \nI, \nC, \nE \rangle$, the {\em $\mR$-model corresponding to~$M$} is the interpretation $\mI = \langle \Delta, \cdot^{\mI}\rangle$ where:
\begin{itemize}
\item $a^\mI = \nI(a)$ for every individual name~$a$,
\item $A^\mI = \{x \in \Delta \mid A \in \nC(x)\}$ for every concept name $A$,
\item $r^\mI = \nE'(r)$ for every role name $r \in \RN$, where $\nE'(R)$ for $R \in \RN \cup \RN^{-}$ are the smallest binary relations on $\Delta$ such that:
  \begin{itemize}
  \item $\nE(R) \subseteq \nE'(R)$, 
  \item if $R \sqsubseteq_\mR S$ then $\nE'(R) \subseteq \nE'(S)$,
  \item if $\Trans{R}$ then $\nE'(R) \circ \nE'(R) \subseteq \nE'(R)$.
  \end{itemize}
\end{itemize}

Note that the smallest binary relations mentioned above always exist: for each $R \in \RN \cup \RN^{-}$, initialize $\nE'(R)$ with $\nE(R)$; then, while one of the above mentioned condition is not satisfied, extend the corresponding $\nE'(R)$ minimally to satisfy the condition.

\begin{lemma} \label{lemma: model graph}
If $\mI$ is the $\mR$-model corresponding to a consistent $\mR$-saturated model graph $\langle \Delta, \nI, \nC, \nE \rangle$, then $\mI$ is a model of $\mR$ and, for every $x \in \Delta$ and $C \in \nC(x)$, we have that $x \in C^\mI$.
\end{lemma}

The first assertion of this lemma clearly holds. The second assertion can be proved by induction on the structure of~$C$ in a straightforward way.

Let $G = \tuple{V,E,\nu}$ be a \CSHIQ-tableau for $\tuple{\mR,\mT,\mA}$ and $v \in V$ be a non-state with $\Status(v) \neq \Unsat$. A {\em saturation path} of $v$ is a sequence $v_0 = v$, $v_1$, \ldots, $v_k$ of nodes of $G$, with $k \geq 1$, such that $\Type(v_k) = \State$ and 
\begin{itemize}
\item for every $0 \leq i < k$, $\Type(v_i) = \NonState$, $\Status(v_i) \neq \Unsat$ and $\tuple{v_i,v_{i+1}} \in E$,  
\item $\Status(v_k) \notin \{\Unsat,\Incomplete\}$. 
\end{itemize}
Observe that each saturation path of $v$ is finite.\footnote{If a non-state $v_{i+1}$ is a successor of a non-state $v_i$ then $\RFormulas(v_{i+1}) \supseteq \RFormulas(v_i)$ and either $\RFormulas(v_{i+1}) \supset \RFormulas(v_i)$ or $\FullLabel(v_{i+1}) \supset \FullLabel(v_i)$. Recall also that $\RFormulas(v_{i+1})$ and $\FullLabel(v_{i+1})$ are subsets of $\closure(\mR,\mT,\mA)$.} Furthermore, if $v_i$ is a non-state with $\Status(v_i) \neq \Unsat$ then $v_i$ has a successor $v_{i+1}$ with $\Status(v_{i+1}) \neq \Unsat$; if $v_i$ is a non-state with only one successor $v_{i+1}$ which is a state then $\Status(v_{i+1}) \neq \Incomplete$, because after a state gets status $\Incomplete$ all edges coming to its are deleted. Therefore, $v$ has at least one saturation path. 

\begin{lemma}[Completeness of \CSHIQ] \label{lemma: completeness}
Let $G = \tuple{V,E,\nu}$ be a \CSHIQ-tableau for $\tuple{\mR,\mT,\mA}$. Suppose $\Status(\nu) \neq \Unsat$. Then $\tuple{\mR,\mT,\mA}$ is satisfiable.
\end{lemma}

\begin{proof}
Let $v_0 = \nu, v_1, \ldots,v_k$ be a saturation path of $\nu$. We define a model graph $M = \langle \Delta, \nI, \nC, \nE\rangle$ as follows:
\begin{enumerate}
\item\label{item: JROSA} Let $\Delta_0$ be the set of all individuals occurring in $\Label(v_k)$ and set $\Delta := \Delta_0$. For every $a \in \Delta_0$, define $\nI(a) = a$. Extend $\nI$ to other individuals from $\IN$ so that: if $(a \doteq b) \in$ $\Label(v_k)$ then $\nI(a) = \nI(b)$; if $a \in \IN$ does not occur in $\tuple{\mR,\mT,\mA}$ then $\nI(a)$ is some individual occurring in $\Label(v_k)$. 
For each $a \in \Delta_0$, mark $a$ as {\em unresolved}$\,$\footnote{Each node of $M$ will be marked either as unresolved or as resolved.} and set $\nC(a) := \{C \mid (a\!:\!C) \in \FullLabel(v_k)\}$.
For each role $R$, set $\nE(R) := \{\tuple{a,b} \mid R(a,b) \in \FullLabel(v_k)\}$.

\item\label{item: JHRES} For every unresolved node $y \in \Delta$ do:
  \begin{enumerate}
  \item If $y \in \Delta_0$ then let $u = v_k$ and $W = \{w \in V \mid \tuple{v_k,w} \in E$ and $\CELabelI(w) = y\}$.
  \item Else let $u = f(y)$ ($f$ is a constructed mapping that associates each node of $M$ not belonging to $\Delta_0$ with a simple state of $G$; as a maintained property of $f$, $\Status(u) \notin \{\Unsat,\Incomplete\}$) and let $W = \{w \in V \mid \tuple{u,w} \in E\}$.

  \item Fix a solution of $\ILConstraints(u)$, and for each $w \in W$:
	\begin{itemize}
	\item if $\CELabelT(w) = \TUS$ then let $n_w = 1$, 
	\item else let $n_w$ be the value of $x_w$ in that solution. 
	\end{itemize}
  \item Delete from $W$ elements $w$ with $n_w = 0$. 

  \item For each $w_0 \in W$ do:
     \begin{itemize}
     \item Let $w_0$, \ldots, $w_h$ be a saturation path of $w_0$ in $G$. 
     \item For $i := 1$ to $n_{w_0}$ do:
	\begin{itemize}
	\item Add a new element $z$ to $\Delta$ and mark $z$ as unresolved.
	\item For each $R \in \CELabelR(w_0)$, add $\tuple{y,z}$ to $\nE(R)$ and $\tuple{z,y}$ to $\nE(R^-)$.
	\item Set $\nC(z)$ to the set of concepts belonging to $\FullLabel(w_h)$ and set $f(z) := w_h$.
	\end{itemize}
     \end{itemize}

  \item Mark $y$ as {\em resolved}.
  \end{enumerate}
\end{enumerate}

The defined model graph $M$ may be infinite. It consists of a finite base created at the step~\ref{item: JROSA} and disjoint trees (with backward edges to predecessors) created at the step~\ref{item: JHRES}.

It is straightforward to prove that $M$ is a consistent $\mR$-saturated model graph. 

By the definition of \CSHIQ-tableaux for $\tuple{\mR,\mT,\mA}$ and the construction of $M$: if $(a\!:\!C) \in \mA$ then $C \in \nC(\nI(a))$; if $R(a,b) \in \mA$ then $\tuple{\nI(a),\nI(b)} \in \nE(R)$; if $(a \not\doteq b) \in \mA$ then $\nI(a) \neq \nI(b)$; and $\mT \subseteq \nC(y)$ for all $y \in \Delta_0$. We also have that $\mT \subseteq \nC(z)$ for all $z \in \Delta - \Delta_0$. Hence, by Lemma~\ref{lemma: model graph}, the interpretation corresponding to $M$ is a model of $\tuple{\mR,\mT,\mA}$.
\myEnd
\end{proof}


\section{Concluding Remarks}
\label{section: conc}

We have developed the first \EXPTIME tableau decision procedure for checking satisfiability of a knowledge base in the DL \SHIQ when numbers are coded in unary.
Our procedure has been designed to increase efficiency of reasoning:
\begin{itemize}
\item We use global state caching but not global caching plus inefficient cuts although the latter approach is much simpler and still guarantees the optimal complexity.

\item We use global state caching but not ``pairwise'' global state caching, cf.\ the lift from anywhere blocking to pairwise anywhere blocking~\cite{HorrocksST00}. This is a good optimization technique, as such a lift would make the graph larger and significantly reduce the chance of getting cache hits. It is a new technique for dealing with both inverse roles and quantified number restrictions. 

\item Similarly to our previous works~\cite{Nguyen-ALCI,SHI-ICCCI,nCPDLreg-long}, but in contrast to~\cite{GoreW09,GoreW10}, if $v$ is a non-state such that $\AfterTrans(v)$ holds then we also apply global caching for the local graph of $v$.

\item Using rules with the highest priority for updating statuses of nodes means that the final statuses $\Unsat$, $\Sat$ and $\Incomplete$ of nodes are propagated ``on-the-fly''.

\item In contrast to Farsiniamarj's method of exploiting integer programming for tableaux~\cite{Farsiniamarj08}, in order to avoid nondeterminism we only check feasibility and do not find and use solutions of the considered set of constraints. Thus, as far as we know, we are the first one who applied integer linear {\em feasibility} checking to tableaux. In the current presentation of our procedure, feasibility checking is done ``on-the-fly''. However, when it turns out that such checks would better be done not ``on-the-fly'', they can be delayed and executed occasionally without affecting soundness and completeness of the calculus (like checking fulfillment of eventualities in the tableau decision procedure for \CPDLreg in~\cite{nCPDLreg-long}). 

\item We do not use pre-compilation techniques. Our operations are direct and natural. (Of course, not all pre-compilation techniques are bad.) The use of NNF for formulas is also natural and can be efficiently handled~\cite{Nguyen08CSP-FI}. Treating TBox axioms as global assumptions is just for making the presentation simple. In practice, the absorption technique like the one discussed in~\cite{NguyenS10TCCI} can be used to deal with TBox axioms. 
\end{itemize}

Our tableau decision procedure for \SHIQ is a framework, which can be implemented with various optimization techniques~\cite{Nguyen08CSP-FI}. 
In~\cite{Nguyen08CSP-FI} we established a set of optimizations that co-operates very well with global caching and various search strategies for the DL \ALC, including formula normalization and caching, literal elimination, propagation of unsat (closedness) for parent and sibling nodes, as well as cutoffs and compacting the graph.\footnote{On small matters like how to implement the procedure $\FindProxy$ efficiently, in~\cite{Nguyen08CSP-FI} we used hashing, which is very efficient. Besides, even in the case the cost of that procedure is polynomial in the size of the graph, the complexity of our tableau decision procedure is still \EXPTIME (when numbers are coded in unary).} All of these optimization techniques can be adapted for \SHIQ, and probably, new ones can be found. 

Implementing an efficient tableau reasoner for \SHIQ is time-consuming. A preliminary implementation of our tableau decision procedure for \SHIQ will be done in the near future. By the preliminary experimental results of \cite{Nguyen08CSP-FI,GorePostniece08} (on global caching for \ALC) and \cite{Farsiniamarj08} (on exploiting integer programming for \SHQ), one can hope that our framework allows to create good reasoners for \SHIQ. 

Our work provides not only techniques for increasing efficiency of reasoning and making it scalable w.r.t.\ quantified number restrictions. It does provide also the first method for developing \EXPTIME tableau decision procedures instead of \NtEXPTIME ones~\cite{HorrocksST00,TobiesThesis} for \EXPTIME description logics with quantified number restrictions (when numbers are coded in unary). The method is also applicable to graded modal logics.


\section*{Acknowledgments} 

This work was supported by Polish National Science Centre grant 2011/01/B/ST6/02759.




\end{document}